\documentclass[letterpaper,11pt]{article}
\usepackage{bm} 
\usepackage[table,xcdraw,svgnames, dvipsnames]{xcolor}
\usepackage{microtype}
\usepackage[utf8]{inputenc}
\usepackage{graphicx}
\usepackage{amsmath, amssymb, amsthm}
\usepackage{comment,hyperref}
\usepackage{thmtools}
\usepackage{tikz}
\usepackage{pgfplots}
\usepackage{varwidth}
\pgfplotsset{compat=1.10}
\usepgfplotslibrary{fillbetween}
\usetikzlibrary{backgrounds}
\usetikzlibrary{patterns}
\usetikzlibrary{positioning}
\usepackage{enumitem}
\usepackage[letterpaper,margin=1in]{geometry}
\usepackage{array}
\usepackage{multirow}
\usepackage{stmaryrd}
\usepackage[font=small]{caption}
\usepackage{makecell}
\usepackage{cite}
\usepackage{booktabs}
\usepackage[nameinlink]{cleveref}
\usepackage[skip=3pt,indent=20pt]{parskip}
\usepackage{mathtools}
\usepackage{natbib}
\usepackage{url}
\usepackage{nicefrac}
\usepackage[algo2e,linesnumbered,ruled,vlined]{algorithm2e}
\usepackage{bbm}
\usepackage{xcolor}

\newcommand{\PP}{\mathbb{P}}

\newcommand{\calA}{\mathcal{A}}

\newcommand{\calT}{\mathcal{T}}
\newcommand{\calE}{\mathcal{E}}

\newcommand{\calB}{\mathcal{B}}

\newcommand{\calD}{\mathcal{D}}

\newcommand{\RR}{\mathbb{R}}
\newcommand{\NN}{\mathbb{N}}
\newcommand{\ZZ}{\mathbb{Z}}

\newcommand{\EE}{\mathbb{E}}
\newcommand{\E}{\mathbb{E}}

\newcommand{\STp}{\normalfont{\text{ST}^{\star}}}

\def\N{\mathbb{N}}
\let\eps\varepsilon
\def\supp{\operatorname{supp}}
\def\set#1{\left\{ #1 \right\}}
\def\prn#1{\left( #1 \right)}
\def\brk#1{\left[ #1 \right]}
\def\midd{\:\middle|\:}
\def\f#1#2{\nicefrac{#1}{#2}}
\def\1{\mathbbm{1}}

\def\argmax{\operatornamewithlimits{arg\,max}}

\newenvironment{algorithm}[2]
{\begin{algorithm2e}[H]
\caption{#1\(\prn{ #2 }\)}
\SetAlgoLined
\SetKwIF{If}{ElseIf}{Else}{if}{then}{else if}{else}{end}
}
{\end{algorithm2e}}

\newtheorem{theorem}{Theorem}
\newtheorem{lemma}{Lemma}
\newtheorem{corollary}{Corollary}

\newtheorem{definition}{Definition}

\newtheorem{proposition}{Proposition}
\newtheorem{observation}{Observation}
\newtheorem{claim}{Claim}

\def\dif#1{\mathop{d #1}}

\def\bE{\ensuremath{\bm{\mathrm{E}}}}

\usepackage{multirow}
\usepackage[noend]{algpseudocode}
\usepackage{algorithm,algorithmicx}

\newlength{\algofontsize}
\hypersetup{
	colorlinks,
	linkcolor={red!50!black},
	citecolor={blue!50!black},
	urlcolor={blue!80!black}
}
\begin{document}
\algrenewcommand\algorithmicrequire{\textbf{Input:}}
\algrenewcommand\algorithmicensure{\textbf{Output:}}
	
\title{Threshold Dynamics and Correlated Prophet Inequalities}
	
\author{
Jos\'{e} Correa
\thanks{Department of Industrial Engineering, Universidad de Chile {\tt (correa@uchile.cl)}}
\and
Maximilian Fichtl
\thanks{TNG Technology Consulting {\tt (mfichtl@gmail.com)}}
\and
Reda Jlibene
\thanks{Moroccan Center for Game Theory, Universit\'{e} Mohammed VI Polytechnique {\tt (reda.jlibene@um6p.ma)}}
\and
Rida Laraki
\thanks{Moroccan Center for Game Theory, Universit\'{e} Mohammed VI Polytechnique {\tt (rida.laraki@um6p.ma)}}
\and
Vasilis Livanos
\thanks{Center for Mathematical Modeling \& Department of Computer Science, University of Southern California {\tt (vas.livanos@gmail.com)}}
\and
Kevin Schewior
\thanks{Department of Mathematics and Computer Science, University of Cologne {\tt (k.schewior@uni-koeln.de)}}
\and
Victor Verdugo
\thanks{Institute for Mathematical and Computational Engineering, and Department of Industrial and Systems Engineering, Pontificia Universidad Católica de Chile {\tt (victor.verdugo@uc.cl)}}\\[1ex]
}
\date{}

\maketitle
\thispagestyle{empty}
\begin{abstract}In recent years, prophet inequalities have become a central tool for analyzing the performance of online algorithms. However, most existing results assume that input random variables are independent, which significantly limits their applicability. Motivated by this gap, we study prophet inequalities under two elementary correlation models induced by a latent \emph{state of the world} variable~$Z$.
In the \emph{common-base model}, the algorithm observes the sequence $Z+X_1,\dots,Z+X_n$, a special case of linear correlations. We analyze single-threshold algorithms with the constraint that they always accept the final item, thereby guaranteeing a reward of at least~$Z$. When $Z$ is chosen adversarially, we characterize the optimal deterministic algorithm of this form, achieving a competitive ratio of~$0.381$. We then show that randomizing the threshold improves the guarantee to~$0.4$, and prove this is optimal via a balanced-prices lower bound. By a minimax argument, the same ratio is achievable when $Z$ is random.

We depart from standard techniques by establishing a stronger lower bound of~$0.41$ and an upper bound of~$0.475$, ruling out the possibility that this class of algorithms attains the $1/2$ ratio known for independent inputs. The core technical contribution is a new analytical framework that captures the reward dynamics of single-threshold algorithms. We introduce a differential equation characterizing the expected reward of a threshold in the worst-case instance, parameterized by the distribution of the maximum. This equation admits a closed-form solution and unifies all known single-threshold prophet inequalities, yielding instance-optimal guarantees and a simple threshold-optimality condition applicable to the common-base model.

Finally, we study the \emph{common-scale model}, where inputs take the form $Z\cdot X_1,\dots,Z\cdot X_n$. We show that even this minimal multiplicative correlation yields strong impossibility results: no algorithm can achieve a competitive ratio exceeding~$1/n$. This is proved via a new notion of multiplicative-invariant stopping times, which also rules out competitive guarantees for secretary-type objectives.
\end{abstract}
\newpage

\section{Introduction}\label{sec:introduction}

Prophet inequalities are one of the most important modern paradigms in online algorithms. In the classic model~\citep{survey18,Hill1982,samuel-84-comparison,Kertz1986}, there are $n$ options, modeled by non-negative random variables $Y_1,\dots,Y_n$, whose realizations are presented to the decision maker one after the other. At time step $i$, the decision maker can either choose to stop, obtaining reward $Y_i$,
or to continue to get a potentially higher reward in the future. The decision maker's goal is to maximize the expected value of their choice. Application areas of this problem include online resource allocation, pricing, matching markets, e-commerce, and hiring processes, which have pushed an extensive research agenda on combinatorial extensions for this problem; see, e.g.,  \cite{lucier2017economic,ChawlaEtAl10,Alaei11,kleinberg2012matroid,FeldmanSZ16,RubinsteinS17,FeldmanGL15,DuettingFKL17,DuttingKL20,CristiCorrea23,CorreaCDHOS23,ezra2020pricing}.

When the variables $Y_i$ are independent and their distributions are known, the seminal result of~\citet{krengel-77-semiamarts,krengel-78-semiamarts} shows that it is possible to achieve a competitive ratio of $1/2$, which is tight. 
The simplest way of achieving this guarantee is by setting a single fixed threshold~\citep{samuel-84-comparison, kleinberg2012matroid} or random threshold~\citep{single-sample-prophet-inequality} and accepting the first value above this bar. 
In this work, we cover all these policies through the lens of differential equations: by understanding the threshold dynamics, we can unify the analyses of all these algorithms and obtain significant new results. In particular, with this new viewpoint, we can address a significant shortcoming of the prophet-inequality literature in light of its practical applications: most works consider only independent random variables, even though most applications involve correlated random variables.
For example, the quality of applicants for a job can depend on whether it is graduation season, the click-through rate of advertisements for sports jerseys may depend on whether the Olympic games are taking place, and the demand for a product in a foreign country may depend on the currency exchange rate. 

To formally describe the setting, let $Y_1, \dots, Y_n$ denote $n$ non-negative, possibly correlated, random variables drawn from a known joint distribution.
We observe $Y_i$ at step $i$ and have to decide immediately and irrevocably whether to select it or not, aiming to maximize the selected value. 
Then, we compare it against a {\it prophet} who can always select the maximum realization. 
We say that an algorithm with stopping time $\tau$ is $\alpha$-competitive, with $\alpha \in [0,1]$, if $\E[Y_{\tau}] \geq \alpha \cdot \E[\max_{i \in [n]} Y_i]$. It is well-known that, for general joint distributions, no non-trivial competitive ratio can be achieved~\citep{hill1992survey}. This motivates the question of what assumptions on the type of correlation allow for deriving better bounds. 
This issue has been partially addressed in the literature; for example, pairwise independent random variables~\citep{caragiannis-pairwise}, linear correlations~\citep{immorlica-linear-correlations}, and Markov random fields~\citep{livanos2024}. 

\subsection{Our Results and Techniques}

The practical applications of prophet inequalities mentioned above have a common theme: the quality of an applicant, and the value of a jersey or product are composed of a common base part, usually representing the state of the environment or context, and an individual part that differs across the options. 
Thus, we consider a nonnegative latent random variable $Z$, representing the base part, and nonnegative random variables $X_1,\dots,X_n$, representing the individual parts. 
All of these variables are independent and have a known distribution with bounded support.\footnote{By standard truncation and normalization techniques, the bounded support assumption is without loss of generality for obtaining prophet inequalities up to vanishingly small factors.}
The variables $Y_i$ that the decision-maker observes then emerge from a simple, known function applied to both $Z$ and $X_i$. The decision-maker only observes $Y_i$; they neither observe $X_i$ nor $Z$ directly.
In their seminal work, \citet[Section 2.2]{milgrom1982theory} use this model to consider information structures in which signals share a common latent state of the world and refer to it as the \emph{common value model} or the \emph{mineral rights model}. This type of correlation also constitutes the canonical example of exchangeable random variables when the $X_i$'s are i.i.d., or conditionally i.i.d.\ given a latent variable~\citep{hewitt1955symmetric}; the i.i.d.\ model has also been studied from an economic point of view since the work of~\citet{wilson77}.
Similar models arise in Kalman filtering and other statistical models; see, e.g.,~\citet{maybeck1982stochastic,laird1982random,carroll1995measurement}.

\paragraph{\bf The common-base model using balanced-prices.}
We first consider the \emph{common-base model}, where $Y_i = Z + X_i$ and assume without loss of generality that all variables are supported in $[0, 1]$.\footnote{This can be done by an appropriate normalization.}
We remark that this setting is covered by the work on linear correlations by~\citet{immorlica-linear-correlations}, which we discuss in detail in~\Cref{sec:related-work}; their work implies a $1/8$-competitive algorithm for the common-base model. 
Note that, since $Z\ge 0$, if the decision maker knows the value of $Z$, they are faced with a standard prophet inequality, and the worst case arises when $Z=0$, where we achieve a factor $1/2$. Our aim is to improve and narrow this competitive ratio; a priori, it is unclear whether one can recover the competitive ratio of $1/2$ above.

When $Z$ is a random variable, a nontrivial tradeoff arises in the prophet inequality. The more values we observe, the better our knowledge of $Z$ and thus the easier it would be to achieve a good competitive ratio in the remainder of the instance. However, as the number of remaining random variables decreases, the expected future value of the instance does so as well. Note that while there is a line of literature considering the common-base model explicitly~\citep{bateni-common-base-value,chawla-base-value}, these papers are not concerned with this tradeoff. Interestingly, we show in this work that a slight variant of single-threshold algorithms can already achieve a good competitive ratio.
Indeed, for the classic independent model, one can generally not hope to improve the competitive ratio of an algorithm by accepting the last value $Y_n$ (in case it has accepted nothing so far) since that value may, in the worst case, be deterministically zero. 
This is different in the common-base model, as one can at least obtain $Z$ by accepting $Y_n$: we set a single threshold, accept the first $Y_i$ above it, or, if there is no such $Y_i$, accept $Y_n$. 
We refer to this class of algorithms as $\STp$.

To analyze this class of algorithms in \Cref{sec:base_value_model}, we apply the minimax theorem, which allows us to consider a randomized threshold and adversarial $Z$ instead. We give a simple analysis showing that randomizing between $\mathbb{E}[\max_{i\in [n]}Y_i]$ (with prob. $0.2$) and $\mathbb{E}[\max_{i\in [n]}Y_i]/2$ (with prob. $0.8$) yields a competitive ratio of $0.4$. This analysis uses the standard balanced-prices bound \citep{samuel-84-comparison,wittmann-threshold,kleinberg2012matroid}, and we also show this factor is best possible for this bound. As a side note, we show that when the threshold is deterministic -- but $Z$ still adversarial -- the best-possible competitive ratio is precisely $(3-\sqrt{5})/2\approx 0.381$.

\paragraph{\bf The threshold dynamics ODE.}
To go beyond $0.4$, we develop a new ordinary differential equation approach to prophet inequalities in~\Cref{sec:ode-approach}. We first illustrate this new technique on the standard prophet inequality. Given some random variables $Y_1,\dots,Y_n$ whose maximum is distributed according to $F$, we obtain an instance with unchanged $F$ (i.e., the distribution of the prophet's value remains the same) but that is worst-possible for single-threshold algorithms (in fact, any online algorithm) as we describe in the following.
Starting from threshold $t=0$,  we define the new instance as we continuously increase $t$. 
To this end, we define precisely the right measure of infinitesimal Bernoulli random variables (i.e., Bernoulli random variables that are $t$ with infinitesimally small probability), so that we recover $F$ as the distribution of the maximum. 
Compared to the original instance, (i) information is only revealed more gradually to the online algorithm by presenting it in infinitesimal slices, and (ii) information is presented in the worst possible order by not giving the algorithm lower fallback options at the end. Therefore, the competitive ratio of threshold algorithms does not improve after this operation.

Interestingly, for the new instance, we determine the {\it threshold dynamics} which describe the behavior of any single-threshold algorithm as a function of the threshold $t$. This yields an ODE parameterized by $F$ that we can solve explicitly. As we demonstrate, this is a powerful tool: We obtain the first unified analysis of all $1/2$-competitive single-threshold algorithms from the literature: the median of $F$~\citep{samuel-84-comparison}, half the expected maximum~\citep{rinott1987comparisons,kleinberg2012matroid}, any value in between these two (folklore), and a random threshold drawn from $F$~\citep{single-sample-prophet-inequality}.
We then apply our new understanding of the threshold dynamics to the common-base model in \Cref{sec:ODEimproving}. Combined with the original approach that yields a competitive ratio of $0.4$, we show that an algorithm from $\STp$ achieves a competitive ratio of $0.41$. While the numerical improvement is small, recall that $0.4$ is the best possible competitive ratio achievable when using standard lower bounds based on balanced prices. Our new method, however, can overcome this barrier.
On the negative side, we establish that no algorithm in $\STp$ achieves a competitive ratio of $0.475$.
We believe that the question of whether there is a $1/2$-competitive algorithm of a different form for the common-base model is a very interesting one for future work.

\paragraph{\bf The common-scale model.} We also consider the case that $Y_i=Z\cdot X_i$, i.e., the correlations are of a very simple non-linear form. 
We refer to this model as the {\it common-scale model}.
When trying to achieve a constant competitive ratio, a natural approach is as follows. Assume the first, say $k$, observations $Y_1,\dots,Y_k$, do not carry a lot of value (indeed, if this were the case, one could blindly select the best of them). Then, one could try finding a good approximation $\tilde z$ to the realization of $Z$, and then running a standard, e.g., single-threshold algorithm on the remaining observations $Y_{k+1},\dots,Y_n$, pretending that the estimation $\tilde z$ is correct. 
Indeed, the set of distributions of $z\cdot X_1,\dots,z\cdot X_n$ for varying $z$ forms a parametric family, so we might hope to use standard statistical methods to estimate the value of $Z$. However, and quite surprisingly, we show in \Cref{sec:upper_bounds} that a competitive ratio better than the trivial $1/n$ is impossible to achieve in general.

We briefly describe the proof idea. We first show that for given random variables $X_1,\dots,X_n$, we can construct a random variable $Z$, such that we can replace any stopping time $\tau$ on $Y_1,\dots,Y_n$ with a stopping time $\tau^{\star}$ that only depends on the realizations of $Y_i/Y_1 = X_i/X_1$. In particular, $\tau^{\star}$ is independent of $Z$, and performs almost as well as $\tau$. Then, we can restate our problem: Given variables $X_1,\dots,X_n$, is there a stopping time $\tau^{\star}$, observing in step $i$ only the quotient $X_i/X_1$, such that $\E[X_{\tau^{\star}}]/\E[\max_{i\in [n]} X_i]$ is uniformly bounded from below by some constant $\alpha$?
We prove that the answer to this question is negative. 

To construct a counterexample $X_1,\dots,X_n$, we proceed in the following way. Each random variable $X_i$ attains the maximum value in its support, say $L^M$, which is equal for all $X_i$, with small but constant probability $\mu$. With probability $1-\delta$, $X_i$ is chosen uniformly from some set $J_i$ of strictly smaller numbers. Now, while $\mu$ is small, if $L$ and $M$ are chosen large enough, the only relevant contribution to $\E[\max_{i\in[n]} X_i]$ comes from instances where one of the variables attains this value $L^M$. Thus, a good stopping time $\tau^{\star}$ would have to be able to detect with high confidence when some variable $X_i=L^M$ is reached. The straightforward approach would be to stop at $i$ whenever the ratio $X_i/X_{i-1}$ becomes exceptionally large. We force this approach to fail by letting the sets $J_i$ be more and more concentrated close to $L^M$. In this way, the ratios $X_i/X_{i-1}$ are almost always large, 
and a jump to $X_i = L^M$ is hardly distinguishable from $X_i$ attaining a value in $J_i$.
We also prove that a related secretary-type problem, where the objective is to stop at the largest $Y_i$ with maximum probability, does not allow for a constant competitive ratio either.
Note that, by taking logarithms, the same negative result holds for the additive secretary setting $Y_i = Z + X_i$. This contrasts with the positive results for the common-base model and the prophet objective. 

\paragraph{\bf Additional results.}
In \Cref{sec:arbitrary_correlations}, instead of assuming that there is an explicit formula expressing the dependencies between variables, we assume that the dependency graph has a simple structure. 
We show that, for arbitrarily correlated instances, there exists an $O(1/\chi)$-competitive prophet inequality, where $\chi$ denotes the chromatic number of the dependency graph.
This implies, by Brooks' theorem \citep{brooks-theorem}, that if every random variable depends on at most $\Delta$ other variables, we can guarantee a competitive ratio of $O(1/\Delta)$, even if this dependence is arbitrary. To the best of our knowledge, this represents the first positive result on correlated prophet inequalities where arbitrary correlation is allowed for subsets of the random variables. One difficulty in proving these results is that the dependency graph captures only pairwise dependencies. While for the prophet objective we can rely on the $1/3$-competitive algorithm by~\citet{caragiannis-pairwise}, we provide a novel analogue result for the secretary setting using a very simple policy.

\subsection{Further Related Work}\label{sec:related-work}

When the random variables are negatively correlated, \citet{rinott1987comparisons} show that the expected reward is at least as high as if the random variables were independent, with the same marginal distributions. Regarding the common-scale model, we remark that the upper bound on the competitive ratio for the classic secretary problem shown by~\citet{ferguson-secretary} is given by a construction similar to ours in the sense that it is of the form $Y_i = ZX_i$, where $Z$ is drawn from a Pareto distribution, and the $X_i$'s are uniformly distributed.

\citet{immorlica-linear-correlations} studied the linear correlation model introduced by~\citet{bateni-common-base-value}: Given a vector $\mathbf X$ of random variables, the vector $\mathbf Y$ of observations is given by $\mathbf Y = A\mathbf X$, where $A$ is a known deterministic matrix. 
Immorlica et al. prove that a competitive ratio of order $O(1/\min\{s_{\text{row}},s_{\text{col}}\})$ can be achieved, where $s_{\text{row}}$ and $s_{\text{col}}$ denote the maximum number of non-zero entries in a column, respectively row, of $A$. They consider not only the single selection prophet setting but also the more general $k$-uniform setting. 
When the column sparsity is bounded, they prove the competitive ratio converges to $1$ as $k$ grows. 
In contrast, a bound only on the row sparsity does not lead to an asymptotic improvement of the competitive ratio.

\citet{caragiannis-pairwise} study posted-price mechanisms and prophet inequalities with pairwise-independent random variables. They prove a $\approx 0.414$ prophet inequality in this case, and they also show a posted-price mechanism that achieves at least a factor of $1/1.299 \approx 0.7698$ of the revenue in the independent case. 
\citet{dughmi2024multiitem} extend some of their results to multiple selection problems, but also derive negative results for general matroids. \citet{anupam-pi-ocrs} follow up on this work and show positive results for various classes of matroids.

In recent works about prophet inequalities with graphical correlations \citep{cai-oikonomou-mrf,livanos2024}, $O(1/\Delta_{\text{MRF}})$-competitive algorithms are derived, where $\Delta_{\text{MRF}}$ is the maximum weighted degree of the Markov random field (MRF). Their model and results might appear similar to ours, but MRFs do not express dependencies via explicit functions, and they are also structurally different from the dependency graph we consider in \Cref{sec:arbitrary_correlations}. 
We remark that our results are orthogonal to theirs, as one can easily construct instances where $\Delta_{\text{MRF}}$ is by an arbitrary factor larger than the standard definition of degree bound $\Delta$ on the independence graph.

\section{Threshold Dynamics}\label{sec:ode-approach}

In this section, we revisit the classic prophet inequality and offer a fresh perspective on it using ordinary differential equations. More precisely, for a fixed distribution of the maximum of random variables, $F$, we derive a differential equation describing the behavior of $V(t)$, the expected performance of a fixed-threshold algorithm with threshold $t$ in the worst possible instance with $F$ as the distribution of the maximum. 

This ODE gives an optimality condition for the best-possible threshold as a function of $F$ and offers a unifying approach to the problem. We recover as simple corollaries the following: (i) The median and half of the expected maximum give a 1/2 prophet inequality~\citep{samuel-84-comparison,rinott1987comparisons}, (ii) by a simple quasi-concavity argument we recover the folklore result that any threshold in between the median and half the expected maximum also gives a 1/2 prophet inequality, (iii) taking a random threshold with distribution $F$ also gives a 1/2 prophet inequality~\citep{single-sample-prophet-inequality}, and (iv) the stronger bound of \citet{hill83} and \citet{samuel-cahn-91-neg-depend} for random variables of bounded support.

Let $F_1, \dots, F_n$ denote the distributions of $X_1, \dots, X_n$, and let $F(t) = \prod_{i = 1}^n F_i(t)$ and $f(t)$ denote the cumulative distribution function and probability density function of $\max_{i\in [n]}X_i$, respectively. 
We assume that $F_1, \dots, F_n$ are continuously differentiable and strictly increasing distributions with support in $[0,1]$, and also that $F(0) = 0$ and $F(1) = 1$.\footnote{By performing standard reductions, these assumptions are w.l.o.g. up to vanishingly small error. These transformations include performing a tiny continuous perturbation of $X_1, \dots, X_n$, as well as truncating and normalizing.}
Now, given $F$, we construct the worst possible instance whose distribution of the maximum is also $F$. To this end, we split the support $[0,1]$ into $m$ smaller sub-intervals for $m > n$ (later, we will take $m \to \infty$). We create a new instance with $m$ Bernoulli random variables $\tilde{X}_{1}, \dots, \tilde{X}_{m}$. The distributions of our new instance are given by $\tilde{X}_i =
{i}/{m}$ w.p. $1-F((i-1)/m)/{F(i/m)}$, and zero otherwise.
Notice that the distribution of $\tilde{X}_i$ depends on the size of the instance $m$; we omit this dependency for notational simplicity.
Next, we observe that for $m$ large enough, the distribution of $\max_{i\in [m]} \tilde{X}_i$ is arbitrarily close to $F$.

\begin{claim}\label{obs:reduction-same-max}
$\max_{i\in [m]} \tilde{X}_i$ converges in distribution to $\max_{i\in [n]} X_i$ as $m\to\infty$.
\end{claim}
\begin{proof}
For any $i \in \{0, 1, \dots, m\}$, we have
\begin{align*}
\Pr\brk{\max_{i\in [m]} \tilde{X}_i \leq \frac{i}{m}} &= \prod_{j = i+1}^m {\prn{1- \frac{F(j/m)-F((j-1)/m)}{F(j/m)}}}= \frac{F(i/m)}{F(1)} = \Pr\brk{\max_{i\in [n]} X_i \leq \frac{i}{m}}.
\end{align*}
The result follows for every value in $[0,1]$ by taking $m \to \infty$.
\end{proof}
We denote by $\tau(t)$ the first random time at which $X_1,\ldots,X_n$ take a value at least $t$, and $\tilde \tau(t)$ is the first random time at which $\tilde X_1,\ldots,\tilde X_m$ take a value at least $t$.
We show that for large enough $m$, the instance $\tilde{X}_{1}, \dots, \tilde{X}_{m}$ is worse for any single-threshold algorithm than the original instance. 
\begin{claim}\label{clm:reduction-is-worse} 
For every $t \in [0,1]$, we have $\E[X_{\tau(t)}]\ge \lim_{m\to \infty}\E[\tilde X_{\tilde \tau(t)}]$.
\end{claim}
\begin{proof}
Note that it suffices to show the claim for $t = i/m$ where $i\in\{0,\dots,m\}$, and consider a value $x = j/m \in [0,1]$ for some $j$.
We will prove that $\Pr[X_{\tau(t)} \geq x]\ge \Pr[\tilde X_{\tilde \tau(t)} \geq x]$ for a sufficiently large $m$. 
Then, by a standard limit density argument, the result holds for every $x\in [0,1]$ and the claim follows by simple integration.
First note that if $x\le t$ then $\Pr[\tilde X_{\tilde \tau(t)} \geq x] = \Pr[X_{\tau(t)} \geq x]$. For $x>t$, observe that for $\tilde X_{\tau(t)}$ to be at least $x$, we need that $\tilde X_i,\ldots, \tilde X_{j-1}$ (all values in $[t,x)$) are zero and at least one of the subsequent values is strictly positive. Thus,
\begin{align*}
\Pr[\tilde X_{\tilde \tau(t)} \geq x] &= \Pr\brk{\tilde{X}_{k} = 0\text{ for }k\in \{i,\ldots,j-1\}} \prn{1 - \Pr\brk{\tilde{X}_{k} = 0\text{ for }k\in \{j,\ldots,m\}}} \\
&= \prod_{k = i}^{j-1} \frac{F((k-1)/m)}{F(k/m)} \prn{1 - \prod_{k = j}^m \frac{F((k-1)/m)}{F(k/m)}} \\
&= \frac{F((i-1)/m)}{F((j-1)/m)} \prn{1 - \frac{F((j-1)/m)}{F(1)}}= \frac{F((i-1)/m)}{F((j-1)/m)} \prn{1 - F\prn{\frac{j-1}{m}}}.
\end{align*}
Then, we have
\[
\Pr\brk{\tilde X_{\tilde \tau(t)} \geq x}
\to \frac{F(t)}{F(x)} (1 - F(x))\text{ as $m\to \infty$}.
\]
Next, by conditioning on the value of the random time $\tau(t)$, we have
\begin{align*}
\Pr[X_{\tau(t)} \geq x] &= \sum_{i = 1}^n \Pr[X_i > x] \prod_{j = 1}^{i-1} \Pr\brk{X_j < t} \\
&= \frac{F(t)}{F(x)} \sum_{i = 1}^n \Pr[X_i > x] \frac{\prod_{k = 1}^n \Pr\brk{X_k < x}}{\prod_{j = i+1}^n \Pr\brk{X_j < t}} \\
&= \frac{F(t)}{F(x)} \sum_{i = 1}^n \frac{\prod_{k = i+1}^n \Pr\brk{X_k < x}}{\prod_{j = i+1}^n \Pr\brk{X_j < t}} \Pr[X_i > x] \prod_{k = 1}^{i-1} \Pr\brk{X_k < x} \\
&> \frac{F(t)}{F(x)} \sum_{i = 1}^n \Pr[X_i > x] \prod_{k = 1}^{i-1} \Pr\brk{X_k < x}\\ 
& = \frac{F(t)}{F(x)} \prn{1 - \prod_{i = 1}^n \Pr[X_i < x]} = \frac{F(t)}{F(x)} (1 - F(x)).
\end{align*}
Here, for the second equality we use that $F(x) = \prod_{i = 1}^n \Pr[X_i < x]$ for all $x \in [0,1]$; for the inequality, we use that $\Pr[X_k < x]>\Pr[X_k < t]$ for $x > t$; and for the following equality that, for any  $p_1, \dots, p_n \in [0,1]$ we have $\sum_{i = 1}^n p_i \prod_{j < i} (1-p_j) = 1 - \prod_{i = 1}^n (1 - p_i)$. 
Since the previous inequality is strict, we obtain that for large enough $m$ it holds $\Pr[X_{\tau(t)} \geq x] \geq \Pr[\tilde X_{\tau(t)} \geq x]$.
\end{proof}
With the previous two results, it is enough to study the performance of single-threshold algorithms on the sequence $\tilde X_1,\ldots,\tilde X_m$ as $m \to \infty$. 
In what follows we use $V(t) = \E[\tilde X_{\tilde \tau(t)}]$.
The following theorem characterizes the dynamics of $V(t)$.
\begin{theorem}\label{thm:v-ode}
The function $V$ is the unique solution of the differential equation
\begin{equation}\label{eq:differential-t}
V'(t) = \prn{V(t) - t} \frac{f(t)}{F(t)}, \qquad V(1)=0.
\end{equation}
Furthermore, its solution is the quasiconcave function
\begin{equation}
\label{eq:v-expression}
V(t) = t - F(t)\left(1 - \int_t^1 \frac{1}{F(u)} \dif u\right),
\end{equation}
which attains its maximum at the unique value $t^*$ such that
\begin{equation}\label{eq:optimal-t-integral}
\int_{t^*}^1 \frac{1}{F(u)} \dif u = 1.
\end{equation}
In particular, $V(t^*)=t^*$.
\end{theorem}
\begin{proof}
We recall that in this instance, for each $i\in \{1,\ldots,m\}$, with $t = i/m$ and $\dif t = 1/m$, we have
$\tilde{X}_i = t$ with probability ${(F(t) - F(t - \dif t))}/F(t)$ and zero otherwise.
Observe that $\tilde X_{\tilde\tau(t)}$ gets value $t$  with probability $(F(t)-F(t-\dif t))/F(t)$, while  with probability $1-(F(t)-F(t-\dif t))/F(t)$, it gets the same as $\tilde X_{\tilde \tau(t+\dif t)}$. Thus,
\begin{align*}
V(t) &= \prn{\frac{F(t)-F(t - \dif t)}{F(t)}} \cdot t + \left(1- \frac{F(t)-F(t - \dif t)}{F(t)}\right) \cdot V(t + \dif t).
\end{align*}
Subtracting $V(t+\dif t)$ and dividing by $\dif t$ we obtain
\begin{align*}
- \frac{V(t + \dif t) + V(t)}{\dif t} &= \frac{F(t) - F(t - \dif t)}{\dif t} \frac{1}{F(t)} \cdot \prn{t - V(t + \dif t)},
\end{align*}
and taking the limit as $\dif t \to 0$ we arrive at the following differential equation.
\begin{equation}
\label{eq:differential-t-inside}
V'(t) = \prn{V(t) - t} \frac{f(t)}{F(t)}, \quad V(1)=0.
\end{equation}
Note that the right-hand side of the differential equation explodes at $t=0$; therefore, we cannot simply take this equation directly in $[0,1]$.
Instead, we analyze the dynamics starting at $\varepsilon>0$ and then define $V$ at zero by taking the limit of $V(\varepsilon)$ when $\varepsilon \to 0$ from the right. 
To this end, consider a fixed value $\varepsilon\in (0,1)$. 
The differential equation \eqref{eq:differential-t-inside} can be written as $V'(t)=g(t,V(t))$, where $g(t,y)=(y-t)f(t)/F(t)$; $g$ is continuous in $t$ and $\max_{t\in [\varepsilon,1]}|\partial_yg(t,y)|\le \max_{t\in [\varepsilon,1]}f(t)/F(t)=L({\varepsilon})<\infty$ since $f(t)/F(t)$ is continuous in the compact $[\varepsilon,1]$; in particular the function $g(t,\cdot)$ is $L({\varepsilon})$-Lipschitz for $t\in [\varepsilon,1]$. 
Then, the Picard-Lindelöf theorem (see, e.g., \cite{teschl2012ordinary}) implies that there is a unique solution to the differential equation \eqref{eq:differential-t-inside}, when $t\in [\varepsilon,1]$.
Moreover, we can get this unique solution to \eqref{eq:differential-t-inside} explicitly. 
Indeed, by simple algebraic manipulation, we get 
\begin{equation*}
V(t) = t - F(t) + \int_t^1 \frac{F(t)}{F(u)} \dif u \text{ for every }t\in (0,1].
\end{equation*}
In what follows, we show that $V$ is quasiconcave in $(0,1)$.
First, observe $V$ is non-negative in $(0,1)$.
Indeed, for every $t>0$, we get
\begin{align*}
V(t)&=t - F(t)\left(1 - \int_t^1 \frac{1}{F(u)} \dif u\right)\ge t - F(t)\left(1 - (1-t)\right)=t(1-F(t))\ge 0.
\end{align*}
Now, for every $t\in (0,1)$ consider the function $r_t(x)=\mathbf{1}_{x\ge t}F(t)/F(x)\le 1$.
For every $x\in (0,1)$, since $F(0)=0$ we have $\lim_{t\to 0}r_t(x)=F(0)/F(x)=0$.
Furthermore, $|r_t(x)|\le F(t)/F(x)\le 1$ for every $x\in (0,1)$.
Then, the dominated convergence theorem implies that 
$$\lim_{t\to 0}\int_{t}^1\frac{F(t)}{F(x)}\dif x=\lim_{t\to 0}\int_{0}^1r_t(x)\dif x\to 0.$$
We conclude that $\lim_{t\to 0}V(t)=-F(0)+0=0$.
Therefore, to show that $V$ is quasiconcave, it is sufficient to prove that $V'$ changes sign exactly once, from positive to negative, on $(0,1)$.
Observe that for every $t\in (0,1)$, we have 
$$V'(t)=1-f(t)+f(t)\int_{t}^1\frac{1}{F(u)}\dif u-F(t)\cdot \frac{1}{F(t)}=f(t)\left(\int_{t}^1\frac{1}{F(u)}\dif u-1\right).$$
As $f$ is strictly positive in $(0,1)$, $V'(t)$ changes sign at most once in the interval $(0,1)$; the latter comes from the fact that $\int_{t}^1(1/F(u))\dif u$ is decreasing as a function of $t$. 
On the other hand, since $\lim_{t\to 0}V(t)=V(1)=0$, Rolle's lemma implies the existence of a point $t^*\in (0,1)$ where $V'(t^*)=0$.
We conclude that $V'$ changes sign exactly once in $(0,1)$ and therefore $V$ is quasiconcave.

Finally,  since $f(t^*)>0$, the condition $V'(t^*)=0$ is equivalent to 
$\int_{t^*}^1 ({1}/{F(u)}) \dif u = 1,$ therefore
$$V(t^*)=t^* - F(t^*)\left(1 - \int_{t^*}^1 \frac{1}{F(u)} \dif u\right)=t^*.$$
In particular, $V$ is non-decreasing in $(0,t^*]$ and non-increasing in $(t^*,1]$, i.e., $V$ attains its maximum at $t^*$.
This completes the proof.
\end{proof}

Theorem \ref{thm:v-ode} reveals, for each fixed distribution $F$ and threshold $t$, the worst-case performance of the corresponding threshold algorithm on an instance having $F$ be the distribution of the maximum value. This is a novel result, and its power will become evident in the following corollaries, which show several well-known results in the area. 
First, from Theorem~\ref{thm:v-ode}, we can immediately obtain the balanced prices lower bound for the expected value of a single-threshold algorithm, first discovered by~\citet{samuel-84-comparison}. This leads to two easy proofs of the result due to \citet{samuel-84-comparison} and \citep{wittmann-threshold,kleinberg2012matroid}.
We call $\bE$ the expectation according to $F$.

\begin{corollary}[Balanced Prices Lower Bound]
\label{cor:sc-lower-bound}
For any $t \in [0,1]$ we have
\begin{equation}\label{eq:sc-lower-bound}
V(t) \geq (1-F(t))t + F(t) (\bE-t).
\end{equation}
Furthermore, for $t_1 = F^{-1}\prn{1/2}$ and $t_2 = \bE/2$, we have $V(t_1) \geq \bE/2$ and $V(t_2) \geq \bE/2$.
\end{corollary}
\begin{proof}
We first rearrange the right-hand side of \eqref{eq:v-expression},
\begin{align*}
t - F(t)\left(1 - \int_t^1 \frac{1}{F(u)} \dif u\right)
&= t\prn{1 - F(t)} + F(t) \prn{t - 1 + \int_t^1 \frac{1}{F(u)} \dif u},
\end{align*}
and note that
\begin{align*}
t - 1 + \int_t^1 \frac{1}{F(u)} \dif u &= \int_t^1 \frac{1 - F(u)}{F(u)} \dif u\\ 
&\geq \int_t^1 \prn{1 - F(u)} \dif u \\
&= \int_0^1 \prn{1 - F(u)} \dif u - \int_0^t \prn{1 - F(u)} \dif u \ge \bE -\int_0^t \dif u = \bE- t,
\end{align*}
from which \eqref{eq:sc-lower-bound} follows. From here we can directly plug in $t_1 = F^{-1}\prn{1/2}$ and $t_2 = \bE/2$ to conclude the last part of the corollary.
\end{proof}

\begin{corollary}[Folklore]
For any $t\in (t_1,t_2)$ we have $V(t) \geq \bE/2$. 
\end{corollary}
\begin{proof}
The quasiconcavity of $V$ implies that for all $a,b$ such that $0\le a\le x\le b\le 1$, we have $V(x)\ge \min\{V(a),V(b)\}$. Then, the result follows from Corollary \ref{cor:sc-lower-bound}.
\end{proof}

Furthermore, our ODE allows us to obtain a third proof of the prophet inequality, which requires only a single sample from each distribution. 
This result is due to~\citet{single-sample-prophet-inequality}.

\begin{corollary}[Single Sample]
\label{cor:single-sample-proof}
Consider $t_3$ drawn according to $F$. 
Then $\E_{t_3 \sim F}[V(t_3)] \geq \bE/2$.
\end{corollary}
\begin{proof}
Note that $(V(t) F(t))' = V'(t) F(t) + V(t) f(t)$, so integrating in $[0,1]$ and using that $F(0) = 0$ and $V(1) = 0$ we obtain
$0 = \int_0^1 V'(t) F(t) \dif t + \int_0^1 V(t) f(t) \dif t.$
Next, we use \eqref{eq:differential-t} to get
\[
0 = \int_0^1 (V(t) f(t) - t f(t)) \dif t + \int_0^1 V(t) f(t) \dif t \iff
0 = 2 \int_0^1 V(t) f(t) \dif t - \int_0^1 t f(t) \dif t. 
\]
Therefore, 
$\int_0^1 V(t) f(t) \dif t = \bE/2$.
Since the left-hand side corresponds to the expectation of the algorithm that selects a threshold $t_3$ drawn at random from $F$, we conclude the result.
\end{proof}

In fact, we can show a stronger claim: suppose that we have $k$ identically distributed draws of the same instance of the prophet inequality problem, with $\bold{X}_j=(X_{1,j},\ldots,X_{n,j})$ the $j$-th draw. Then, a single threshold $t$ drawn according to $F^k$ is $1/(k+1)$-competitive with respect to the expected maximum of the $k$ draws. 
Since $t^*$ maximizes $V$, it achieves a $1/(k+1)$-competitive ratio simultaneously for all $k$.

\begin{corollary}
\label{cor:k-instances-proof}
Consider $t_4$ drawn according to $F^k$. Then, we have $\E_{t_4 \sim F^k}[V(t_4)] \geq \bE_k/(k+1)$, where $\bE_k = \E[\max_{j\in [k], i\in [n]}X_{i,j}]$.
\end{corollary}
\begin{proof}
By differentiating the product $V(t) F^k(t)$ and using the differential equation \eqref{eq:differential-t}, we get
\begin{align*}
\frac{\dif{}}{\dif t}\prn{V(t) F^k(t)} &= V'(t) F^k(t) + k V(t) F^{k-1}(t) f(t) \\
&= ((k+1)V(t) - t) f(t) F^{k-1}(t).
\end{align*}
By integrating in $[0,1]$ and using that $F(0) = V(1) = 0$, we obtain
\[
(k+1) \int_0^1 V(t) f(t) F^{k-1}(t) \dif t = \int_0^1 t f(t) F^{k-1}(t) \dif t,
\]
and multiplying both sides by $k$ yields
\[
(k+1) \int_0^1 V(t) k f(t) F^{k-1}(t) \dif t = \int_0^1 t k f(t) F^{k-1}(t) \dif t.
\]
Since the left-hand side corresponds to the expectation of the algorithm that selects a threshold $t_4$ drawn at random from $F^k$ and the right-hand side corresponds to $\bE_k$, the result follows.
\end{proof}

In addition, we can recover a result by \citet{hill83}, originally proved for the optimal stopping rule and later for threshold policies by~\citet{samuel-cahn-91-neg-depend}. 
\begin{corollary}[Improved Inequality in the Bounded Case]
\label{cor:hill-stronger-bound}
For the $t^*$ that maximizes $V$, we have
$\bE \leq 2t^* - {t^*}^2.$
In particular, $V({t^*}) \geq 1 - \sqrt{1 - \bE}$.
\end{corollary}
\begin{proof}
Recall that the optimal $t^*$ that maximizes $V$, by \eqref{eq:optimal-t-integral} satisfies $\int_{t^*}^1 ({1}/{F(u)}) \dif u = 1$. Thus,
\begin{align*}
\int_0^1 F(x) \dif x &\geq \int_{t^*}^1 F(x) \dif x\\
&= \prn{\int_{t^*}^1 \frac{1}{F(x)} \dif x} \prn{\int_{t^*}^1 F(x) \dif x} \geq \prn{\int_{t^*}^1 \frac{1}{\sqrt{F(x)}} \cdot \sqrt{F(x)} \dif x}^2= (1 - t^*)^2,
\end{align*}
where the second inequality follows by using the Cauchy–Schwarz inequality. Thus,
\begin{align*}
\bE = \int_0^1 \prn{1 - F(x)} \dif x = 1 - \int_0^1 F(x) \dif x \leq 1 - (1-t^*)^2 = 2t^* - {t^*}^2.
\end{align*}
Furthermore, recall that $t^* = V(t^*)$. Since $\bE,t^*\in [0,1]$ we conclude 
$V(t^*) = t^* \geq 1 - \sqrt{1 - \bE}$.
\end{proof}
The previous corollary provides a stronger prophet inequality for random variables supported in $[0,1]$. This can be easily extended to bounded-support random variables by normalization.
Finally, our result recovers the $(1/4)$-\emph{difference prophet inequality} by \citet{hill-kertz-difference-pi}.
\begin{corollary}[Difference Prophet Inequality]
\label{cor:hill-kertz-difference-pi}
For the $t^*$ that maximizes $V$ we have $\bE - V(t^*) \leq 1/4$.
\end{corollary}
\begin{proof}
Using the fact that $V(t^*) = t^*$ in Corollary~\ref{cor:hill-stronger-bound} we get that $\bE - V(t^*) \leq t^* - {t^*}^2$. Since $x - x^2 \leq 1/4$, the result follows.
\end{proof}

\paragraph{\bf Further consequences.} One crucial aspect of our threshold dynamics analysis is that the reward function $V$ in \eqref{eq:v-expression} is parameterized by $F$, and therefore, instance dependent. 
This allows us to obtain instance-optimal prophet inequalities for restricted families of distributions for the maximum. 
More specifically, suppose we want to find the worst-case approximation factor of single-threshold policies over a family of distributions $\mathcal{F}$, e.g., distributions with monotone hazard rate (MHR) or bounded second moment.
By Theorem \eqref{thm:v-ode}, the worst-case approximation ratio of a single-threshold policy for $\mathcal{F}$ can be obtained by the following scheme. 
Given $\alpha\ge 0$, let 
$$\phi(\alpha)=\inf\left\{t-\alpha \int_{0}^1(1-F(u))du:1=\int_t^1\frac{1}{F(u)}du\text{ and }F\in \mathcal{F}\right\}.$$
The worst-case approximation factor over $\mathcal{F}$ is given by the largest value $\alpha$ such that $\phi(\alpha)\ge 0$.
We remark that for each $\alpha$, the infinite-dimensional optimization problem that defines $\phi(\alpha)$ can be formulated under specific knowledge about $\mathcal{F}$.
For instance, when $\mathcal{F}$ is the MHR family, we can parameterize it in terms of the hazard rate function to get a variational problem.

\section{The Common-Base Model}\label{sec:base_value_model}

Recall that in the common-base model, there are $n+1$ non-negative and independent random variables $X_1, \dots, X_n$ and $Z$, and the decision-maker sequentially observes $Y_i = X_i + Z$ for $i = 1,\dots, n$. In this special case, the bound of \citet{immorlica-linear-correlations} evaluates to $1 / (4e^3)$, though it can be improved to $1/8$. The algorithm in \cite{immorlica-linear-correlations} randomizes between two strategies with equal probability: accepting $Z + X_1$ whenever it is at least $\E[Z] / 2$ (and not accepting anything otherwise), or using the threshold $\E[\max_i X_i]/2$. However, it remained unclear whether such performance could be matched—or even improved—by a deterministic, single-threshold rule, and whether the optimal $1/2$ competitive ratio known for the classical prophet inequality could still be achieved in the common-base model. This question essentially asks whether the guarantees that hold under independence extend to settings with even the weakest forms of correlation.

In fact, it is possible to improve upon the $1/8$ competitive ratio while still using a single-threshold algorithm. As we show in~\Cref{app:13single}, setting the standard threshold $T = \E[\max_{i \in [n]} X_i]/2$ yields a competitive ratio of $1/3$ in the common-base model with $Y_i = Z + X_i$ (\Cref{prop:one_third_sum}). However, we also resolve the preceding question in the negative: we show that no single-threshold algorithm can achieve a competitive ratio exceeding $0.4765$ in the common-base model, thus establishing a clear separation between the common-base model and the fully independent prophet inequality.
Indeed, to show this, let $Z$ be distributed according to
$\Pr[Z \leq x] = x^{0.35}$
for $x \in [0,1]$. For $i\in \{1,\dots,n-1\}$, set
$X_i = {i}/{(n-1)},$
with probability $1$, and let $X_n$ be $0.85/\eps$ w.p. $\eps$, and $0$ w.p. $1 - \eps$.
The prophet's value is given by
\[
\E[Z+\textstyle\max_{i\in [n]} X_i] = \int_0^1 \Big(1-x^{0.35}\Big) dx + (1-\eps)+ 0.85,
\]
which converges to $0.35/1.35 + 1.85 \approx 2.1092$ as $\eps \rightarrow 0$.

Now consider any threshold $T$. When $T \leq 1$, we either stop at $X_1+Z$, if $Z \geq T-1/(n-1)$, or at the first index $i$ such that $Z + i/(n-1) \geq T$. Since $Z \leq 1$ and $T \leq 1$, we receive no more than $1+1/(n-1) \rightarrow 1$ for $n\rightarrow \infty$. 
When $T\geq 2$, we always continue to the last variable $X_n$, and obtain $\eps\cdot(0.85/\eps + \E[Z]) \rightarrow 0.85$ as $\eps\rightarrow 0$. Now suppose that $1 \leq T \leq 2$. Then for $n\rightarrow \infty$ and $\eps \rightarrow 0$, the algorithm gets at most
\begin{equation*}
\Pr[Z + 1 \geq T]T + 0.85\cdot \Pr[Z+1 < T]= (1-(T-1)^{0.35})T + 0.85(T-1)^{0.35}.
\end{equation*}
This expression has a maximum value of $\approx 1.00474$ for $T \in [1,2]$. Hence, the competitive ratio of this instance is approximately
${1.00474}/{2.1092} \approx 0.4764$.

In~\Cref{subsec:lastvalue} we show how to enhance the performance of single-threshold algorithms by implementing a simple modification: if no value is selected by the end, we choose the last one for sure. We show the following lower bounds on the competitive ratio: 0.4 is achievable by randomized single-threshold policies (\Cref{prop:ST+LB}) while $\approx 0.381$ is attained by deterministic single-threshold policies (\Cref{prop:adv-golden}). Finally, in \Cref{sec:ODEimproving}, we further improve by showcasing the threshold-dynamics machinery to get a lower bound of 0.41 (\Cref{prop:improvedLB}).
\subsection{(1/3)-Competitive Single-threshold Algorithm}\label{app:13single}

In the following proposition, we show that the standard threshold $T = \E[\max_{i \in [n]} X_i]/2$ yields a competitive ratio of $1/3$ in the common-base model with $Y_i = Z + X_i$ for every $i\in [n]$.

\begin{proposition}\label{prop:one_third_sum}
Let $Z$ and $X_1,\dots,X_n$ be independent random variables. Let $\tau$ be the stopping time accepting the first observation $Y_i=Z+X_i$ attaining at least $T = \E[\max_{i\in [n]} X_i]/2$. Then
\begin{align*}
    \frac{\E[Y_{\tau}]}{\E[\textstyle\max_{i\in [n]} Y_i]} \geq \frac{1}{3}.
\end{align*}
\end{proposition}

\begin{proof}[Proof]
Let $W_1,\dots,W_{2n}$ be the random variables defined as follows: for each $i\in\{1,\dots,n\}$, $W_{i}$ is distributed according to
    \[
    \Pr[W_i \leq x] = \begin{cases}
    \Pr[X_i \leq x \mid X_i < T] & \text{ if }\Pr[X_i<T]>0,\\
    1 & \text{ otherwise,}
    \end{cases}
    \]
    and $W_{n+i}$ is such that
    \[
    \Pr[W_{n+i} \leq x] = \begin{cases}
    \Pr[X_i < T] &\text{ if } 0 \leq x < T \\
    \Pr[X_i \leq x] &\text{ if } x \geq T.
    \end{cases}
    \]
    Let $Y^X_i = Z+X_i$ for each $i\in [n]$ and $Y^W_j = Z+W_j$ for each $j\in [2n]$. 
\begin{claim}\label{lem:variable_decomposition}
    $\E[\max_{i\in [n]} Y^X_i] = \E[\max_{j\in [2n]} Y^W_j]$, and for any stopping rule $\tau$, $\E[Y^X_{\tau}] \geq \E[Y^W_{\tau}]$.
\end{claim}
Thanks to \Cref{lem:variable_decomposition}, it suffices to show that
${\E[Y^W_{\tau}]}/{\E[\max_{i\in [n]} Y^W_i]} \geq {1}/{3}.$
Let $\tau$ be the threshold stopping rule applied to the variables $Y^W_j$, and $\rho$ the threshold stopping rule applied to the variables $W_j$, both with respect to the threshold $T$.
We distinguish three cases.
\begin{enumerate}[itemsep=0pt,label=\normalfont(\arabic*)]
    \item $Z \geq T$. Then $\tau = 1$, and we obtain a profit of $Z+W_1 \geq Z$.
    \item $Z < T$ and $\tau \in [n]$. In particular, the algorithm stops, so we get at least $T$.
    \item $Z < T$ and $Y^W_i = Z+W_i < T$ for all $i \in [n]$. Then $\tau$ stops at the first index $n+i$ with $W_{n+i} \neq 0$ (if such variable exists), and so does $\rho$. Thus, $\tau=\rho$ and in particular $Y^W_{\tau} \geq W_{\rho}$.
\end{enumerate}
Define the event $E = \{Z+X_i < T\;\text{for all }i\in [n]\}$. Since $\rho$ never stops at the first $n$ random variables $W_j$ and does not depend on $Z$, $\rho$ is independent of $\{Z < T\} \cap E$. Thus, from our considerations above, we deduce
\begin{align*}
    \E[Y^W_{\tau}] \geq \Pr[Z \geq T]\E[Z \mid Z\geq T] + \Pr(\{Z < T\} \cap E^c)T + \Pr[\{Z < T\} \cap E]\E[W_{\rho} \mid Z<T \cap E].
\end{align*}
Using independence, we get
$\E[W_{\rho} \mid Z<T \cap E] = \E[W_{\rho}] \geq T,$
where we use the classical prophet inequality result that the threshold $T = \E[\max_{j\in [2n]} W_j]/2$ gives an expected reward of at least $T$.
We conclude that $\E[Y^W_{\tau}] \geq \Pr[Z \geq T]\E[Z \mid Z\geq T] + \Pr[Z < T]T = \E[\max\{Z,T\}].$
Finally, since
\[
\max\{Z,T\} \geq \frac{1}{3}Z+\frac{2}{3}T = \frac{1}{3}\left(Z+\E[\textstyle\max_{i \in [n]} W_i]\right),
\]
we arrive at the result.

Now we prove \Cref{lem:variable_decomposition}.
By construction, for every $i\in [n]$, the random variables $\max\{W_{i},W_{n+i}\}$ and $X_i$ are equally distributed: for $x \geq T$, $\Pr[W_i \leq x] = 1$, so
    $\Pr[X_i \leq x] =\Pr[W_{n+i} \leq x] = \Pr[\max\{W_i,W_{n+i}\} \leq x].$
    On the other hand, for $x < T$,
    \begin{align*}
    \Pr[X_i \leq x] &= \Pr[X_i \leq x \wedge X_i < T] \\
    &= \Pr[X_i \leq x \mid X_i<T]\Pr[X_i < T] \\
    &= \Pr[W_i \leq x]\Pr[W_{n+i}\leq x] = \Pr[\max \{W_i,W_{n+i}\} \leq x].
    \end{align*}
    Thus, $\E[\max_{i\in [n]} X_i] = \E[\max_{j\in [2n]} W_j]$, so also $\E[\max_{i\in [n]} Y^X_i] = \E[\max_{i\in [n]} Y^W_i]$.
    
    Let us consider the stopping rule $\tau$ on the $n$ variables $\widetilde Y_i = Z+\max\{W_i,W_{i+1}\}$ and let $\rho$ be the stopping rule on the $2n$ variables $Y^W_j$ with respect to the same threshold as $\tau$. 
    Whenever $\rho \in\{i,n+i\}$, we have $\tau = i$. Thus, $W_{\rho} \leq \max \{W_{\tau},W_{\tau+n}\}$ and consequently $Y^W_{\rho} \leq \widetilde Y_{\tau}$. This proves that
    $\E[Y^W_{\rho}] \leq \E[\widetilde Y_{\tau}] = \E[Y^X_{\tau}],$
    where in the last equality we use again that $\max\{W_i,W_{n+i}\}$ and $X_i$ are identically distributed.
\end{proof}

\subsection{The Advantage of Selecting the Last Value}\label{subsec:lastvalue}

In this section, we demonstrate how to enhance the performance of single-threshold algorithms by implementing a simple feature: if no value is selected by the end, we choose the last one for sure. During this section, we denote by $V_{X}(t)$ the expected reward of the single-threshold algorithm that sees the instance $X_1,\ldots,X_n$ and selects the first element with value at least $t$, and we denote by $\hat V_{Y}(t)$ the expected reward of the single-threshold algorithm that sees the instance $Y_1,\ldots,Y_n$, with $Y_i=Z+X_i$ for each $i\in [n]$, and selects the first element with value at least $t$ or the last value otherwise.
For simplicity of exposition, we sometimes refer to these policies as $\STp$.
The following observation is relevant to the analysis we perform in this section.
\begin{observation}\label{obs:decomposition}
$\hat V_{Y}(t)=\EE[Z]+\EE[V_X(t-Z)]$.
\end{observation}

The proof follows by noting that for any value $z$ we have that $\hat V_{Y}(t)=z+\EE[V_X(t-z)]$.
\Cref{obs:decomposition} and the balanced-prices lower bound for $V_X$ (\Cref{cor:sc-lower-bound}) imply that for $t\in [0,1]$ we have $\EE[V_X(t-Z)]\ge \EE[\max(0,\min(t-Z,\bE-(t-Z)))]$, where $\bE=\EE[\max_{i\in [n]}X_i]$. 
Therefore,
\begin{equation}
\hat V_Y(t)\ge \EE[Z]+\EE[\max(0,\min(t-Z,\bE-t+Z))]\ge \EE[\min(t,\bE-t+2Z)].\label{LB:STp}
\end{equation}

\paragraph{\bf Randomized policies and adversarial setting.}
To showcase the advantage of the $\STp$ policies over usual single-threshold rules, we show that the competitive ratio of the $\STp$ family is, in fact, at least $0.4$ if we are allowed to take a {\it random} single threshold.
The analysis is simple if we take a detour through the {\it adversarial} version of the common-base model. 
In this variant, the goal is to find a randomized single-threshold policy in $\STp$ such that, when the adversary then chooses a deterministic value $Z\in [0,1]$, the competitive ratio is maximized.
An application of the minimax theorem implies that the worst-case competitive ratio achievable by randomized single-threshold policies in $\STp$ for the adversarial model (i.e., deterministic $Z$) is equal to the worst-case competitive ratio achievable by deterministic policies in $\STp$ in the common-base model (i.e., random $Z$).

\begin{proposition}\label{prop:ST+LB}
In the common-base model, the competitive ratio of the $\STp$ family is at least $0.4$.
\end{proposition}

\begin{proof}
By the previous discussion, to prove the $0.4$ lower bound, it is sufficient to focus on randomized single-threshold policies in $\STp$.
Consider a random threshold $T$ such that it is equal to $\bE$ with probability 0.2, and equal to $\bE/2$ with probability 0.8.
Observe that the bound in \eqref{LB:STp} in particular holds when $Z$ is deterministic and equal to $z$, i.e., in the adversarial common-base model.
Therefore,
    \begin{align*}
        \EE[\hat V_Y(T)] &\ge z+0.8\max(0,\min(\bE/2-z,\bE-\bE/2+z))+0.2\min(\bE-z,z)\\
        &= z+0.8\max(0,\bE/2-z)+0.2\min(\bE-z,z)\\
        &= z+0.8(\bE/2-z)\mathbf{1}_{z \leq \bE/2}+0.2\cdot z \mathbf{1}_{z \leq \bE/2}+0.2(\bE-z)\mathbf{1}_{z \geq \bE/2}\\
        &= ( 0.8(\bE/2 - z) + 1.2z ) \mathbf{1}_{z \leq \bE/2} +  (0.2\cdot \bE+0.8\cdot z) \mathbf{1}_{z \geq \bE/2} \\
        &= 0.4(\bE+z) \mathbf{1}_{z \leq \bE/2} + (0.2\cdot \bE+0.8\cdot z) \mathbf{1}_{z \geq \bE/2}\\
        &\geq 0.4(\bE+z) \mathbf{1}_{z \leq \bE/2} + (0.4\cdot \bE+0.4\cdot z) \mathbf{1}_{z \geq \bE/2}= 0.4(\bE + z),
    \end{align*}
where the first inequality holds from applying \eqref{LB:STp} in each of the two possible realizations of $T$, the third equality follows from decomposing $z=z\mathbf{1}_{z \leq \bE/2}+z\mathbf{1}_{z \geq \bE/2}$, and the last inequality holds since, for $z\ge \bE/2$, we have $0.2\cdot \bE+0.8\cdot z\ge 0.4\cdot \bE+0.4\cdot z$.    
\end{proof}
Interestingly, the 0.4 bound in \Cref{prop:ST+LB} is tight for the balanced prices bound in \eqref{LB:STp}, i.e., this approximation factor is optimal for the analysis we provide.
\begin{proposition}\label{prop:tightSC}
For every $\delta>0$ there exists an instance $X_1,\ldots,X_n,Z$ of the common-base model where 
$\EE[Z]+\EE[\max(0,\min(t-Z,\bE-t+Z))]\le (0.4+\delta)\EE[Z+\textstyle\max_{i\in [n]}X_i]$ for every $t\ge 0$.
\end{proposition}
\begin{proof}
Fix a value $\varepsilon>0$.
For $i\in \{1,\dots,n-1\}$, set $X_i = {i}/{(n-1)}$
with probability $1$, and let $X_n$ be $1/\varepsilon$ with probability $\varepsilon$, and zero with probability $1 - \eps$.
In particular, $\bE=2-\varepsilon$. Now consider a random variable $Z$ taking the value $0$ with probability $0.5$ and $\bE/2=1-\varepsilon/2$ with probability $0.5$, i.e., $\EE[Z]=0.5-\varepsilon/4$.
Then, 
\begin{align*}
&\EE[Z]+\EE[\max(0,\min(t-Z,\bE-t+Z))]\\
&=0.5-\varepsilon/4+0.5\cdot \max(0,\min(t,\bE-t))+0.5\cdot \max(0,\min(t-\bE/2,(3/2)\bE-t))\\
&=0.5-\varepsilon/4+0.5\cdot \max(0,\min(t,2-\varepsilon-t))+0.5\cdot \max(0,\min(t-1+\varepsilon/2,3-\varepsilon/2-t)).
\end{align*}
The previous function is upper bounded by $1+O(\varepsilon)$ for $t\ge 0$, and $\EE[Z+\max_{i\in [n]}X_i]=2.5-(5/4)\varepsilon$.
Then, the ratio between the balanced prices bound and the prophet's value goes to $0.4$ as $\varepsilon \to 0$.
\end{proof}

We also show that in the adversarial common-base model, there is a strict separation in the power of using randomized single-threshold policies.
Namely, if we restrict to deterministic policies in $\STp$, the competitive ratio of this family is strictly less than 0.4, and in fact, equal to $(3-\sqrt{5})/2\approx 0.381$.
\begin{proposition}\label{prop:adv-golden}
In the adversarial common-base model, the competitive ratio of deterministic policies in $\STp$ is equal to $(3-\sqrt{5})/2\approx 0.381$.
\end{proposition}
\begin{proof}
We start by showing the lower bound. 
For notational simplicity, let's denote $(3-\sqrt{5})/2$ by $\alpha$.
To lower-bound the competitive ratio by $\alpha$, it is sufficient to prove the existence of a constant $\beta>0$ satisfying the following inequality:
\begin{equation}
\min(\beta\bE-\alpha z,\bE-\beta\bE+(2-\alpha)z)\ge \alpha\bE \quad \text{for every }z\in [0,1].\label{LB:suff-adv}
\end{equation}
Indeed, by the lower bound in \eqref{LB:STp}, the previous inequality implies that for the single-threshold policy in $\STp$ with $t=\beta\bE$ we have $\hat V_Y(\beta\bE)\ge \alpha \bE+\alpha z=\alpha (\bE+z)$, i.e., its competitive ratio is lower-bounded by $\alpha$.
We now prove \eqref{LB:suff-adv}.
Consider the concave function $g(z)=\min(\beta \bE-\alpha z,\bE-\beta \bE+(2-\alpha)z)$ for $z\in [0,\beta \bE]$.
The minimum value is attained at the endpoints, that is, with $g(0)=\bE\min(\beta,1-\beta)$ or $g(\beta \bE)=\beta \bE\min(1-\alpha,1/\beta+3-\alpha)=\beta \bE (1-\alpha)$, and so $g(z)\ge \bE \min(\min(\beta,1-\beta),1-\alpha)$.
Then, it is sufficient to take $\beta=1-\alpha\approx 0.618$ since in that case $g(z)\ge \bE\min(\alpha,1-\alpha)=\alpha \bE$.

To prove the upper bound, given $\varepsilon>0$, consider the following instance: two deterministic random variables $X_1=0$ and $X_2=(1+\sqrt{5})/2$, and a third random variable $X_3$ equal to $1/\varepsilon$ with probability $\varepsilon$ and equal to zero with probability $1-\varepsilon$, i.e., $\EE[X_3]=1$. 
If the threshold $t$ is smaller than or equal to $(1+\sqrt{5})/2$, the adversary selects $z=t$, and therefore the decision-maker stops at the first value to get $z=t\le (1+\sqrt{5})/2$.
In this case, the competitive ratio is at most $t/(t+(1+\sqrt{5})/2+1)\le (3-\sqrt{5})/2$, where the inequality holds since $t^2\le t+1$.
If the threshold $t$ is larger than $(1+\sqrt{5})/2$, the adversary selects $z=0$, and the decision-maker stops at the third value, getting, on expectation, a value of $1$.
In this case, the competitive ratio is at most $1/((1+\sqrt{5})/2+1)= (3-\sqrt{5})/2$.
Then, the competitive ratio is no larger than $(3-\sqrt{5})/2$.
\end{proof}

\paragraph{\bf Upper bound for $\STp$ policies.} Finally, we show that for the common-base model, the competitive ratio of the policies in $\STp$ is actually less than $0.475$, which improves on the (simpler) upper bound of $0.4765$ stated at the beginning of this section. 
\begin{proposition}\label{prop:UBSTp}
In the common-base model, the competitive ratio in $\STp$ is less than $0.475$. 
\end{proposition}
To prove the upper bound, we use a sequence of hard instances consisting of $2N$ values each, i.e., $Y_i=Z+X_i$ for each $i\in \{1,\ldots,2N\}$, and then we perform a limit analysis as $N$ grows large.
Let $N \in \mathbb{N}$ and $\beta\in (0,1)$. 
Consider $t\mapsto q(t)=\beta/(t-1-\beta^2)$ and $p_i=q(i/N)/N$ for each integer $i>(1+\beta^2)N+1$.
Define the random variables $X_1,\ldots,X_{2N}$ as follows: a) for every $1 \leq i \leq (1+\beta^2) N+1$ we have $X_i ={i}/{N}$, b) for every $(1+\beta^2)N+1 < i <2N$ we have $X_i=i/N$ with probability $p_i$ and zero otherwise, and c) $X_{2N}=(1-\beta)N$ with probability $1/N$ and zero otherwise.
In the following lemma, we verify that the expectation of the maximum value in the sequence $X_1,\ldots,X_{2N}$ is approximately $2$.

\begin{lemma}\label{lem:UBstep1}
$\EE[\max_{i\in \{1,\ldots,2N\}}X_i]=2+o(1)$.
\end{lemma}
\begin{proof}
We start by studying the distribution of $M=\max_{i\in \{1,\ldots,2N\}}X_i$. 
If $t \leq 1+\beta^2$, then $\PP[M>t] = 1$ since $M \geq \lfloor(1+\beta^2)N+1\rfloor/N> 1+\beta^2\ge t$.
If $t \geq 2$, then $\PP[M> t] = {1}/{N}$ for $N$ sufficiently large, thanks to the random variable $X_{2N}$.
If $1+\beta^2 < t \leq 2$, then
$$\PP[M\le t] = \left(1-\frac{1}{N}\right)\prod_{s=\lceil Nt \rceil+1}^{2N-1}(1- p_s),$$
i.e., $X_i=0$ for every $i\ge \lceil Nt\rceil+1$.
Then, we have
    \begin{align*}
    \ln(\PP[M\le t])& = \ln\left(1-\frac{1}{N}\right)+ \sum_{s=\lceil Nt \rceil +1}^{2N-1} \ln (1-p_s)\\
    &=\sum_{s=\lceil Nt \rceil }^{2N-1} \left(-\frac{q({s}/{N})}{N}+O({N^{-2}})\right) \\
    & = - \int_t^2 q(s) \dif s + o(1)= -\beta \ln\left(\frac{1-\beta^2}{t-1-\beta^2}\right) + o(1),
    \end{align*}
    and therefore
    \begin{align*}
    \PP[M > t] &= 1-\left(\frac{t-1-\beta^2}{1-\beta^2}\right)^\beta+o(1).
    \end{align*}
    Hence, overall, we have 
    \begin{align*}
        \EE[M] &= \int_0^{\infty} \PP[M>t]\dif t \\
        &= o(1)+1+\beta^2 + \int_{1+\beta^2}^2\left(1-\left(\frac{t-1-\beta^2}{1-\beta^2}\right)^{\beta}\right)\dif t+\frac{(1-\beta)N-2}{N}\\
        &= o(1)+1+\beta^2+\beta(1-\beta)+1-\beta= o(1)+2.\qedhere
    \end{align*}
\end{proof}
Consider the continuous piecewise function $h:\RR\to \RR$ defined as follows: 
$$
h(t)=
\begin{cases}
        t &\text{if $ t \leq 1+\beta^2$,} \\
        \displaystyle\frac{1+\beta^2-\beta t}{1-\beta}  &\text{if $1+\beta^2 \leq t\leq 2$,} \\
        1-\beta &\text{if $t \geq 2$.}
    \end{cases}$$
Consider the common-base instance $Y_i=X_i+Z$ for each $i\in \{1,\ldots,2N\}$, where $Z=1-\beta^2$ with probability $\beta\in (0,1)$, and is equal to zero with probability $1-\beta$. In the following lemma, we compute the expected reward of the $\STp$ policy with threshold $t$ as a function of $t$ and $\beta$, as $N$ grows large.
\begin{lemma}\label{lem:UBstep2}
For every $t\in (0,\infty)$, we have $\hat V_{Y}(t)\to \beta(1-\beta^2)+\EE[h(t-Z)]$ as $N\to \infty$.
\end{lemma}
\begin{proof}
From \Cref{obs:decomposition}, for every $N$ we have $\hat V_Y(t)=\EE[Z]+\EE[V_X(t-Z)]$, and $\EE[Z]=\beta(1-\beta^2)$.
Therefore, it suffices to show that $\EE[V_X(t-Z)]\to \EE[h(t-Z)]$ as $N\to \infty$. Now we analyze the value $V_X(s)$ for each $s$.
When $s \in [0 , 1+\beta^2]$, the policy stops at $X_{\lceil Ns \rceil}$, thus $V_X(s) = {\lceil Ns \rceil}/{N} \rightarrow s = h(s)$ when $N\to \infty$, and
when $s \geq 2$ we get $(1-\beta)N$ w.p. $1/N$, so $V_X(s) = 1-\beta=h(s)$ as $N\to \infty$.

Consider now $s \in (1+\beta^2,2)$, and suppose $sN \in \mathbb{N}$; the general case follows from a standard density argument.
Recall that $X_{Ns}= s$ with probability $p_{ Ns}={q(s)}/{N}$, and then in this event $V_X(s)$ is $s$, otherwise $X_{Ns}=0$ and $V_X(s)$ yields $V_X({s+{1}/{N}})$. 
Hence $V_X$ obeys the following recursion with a boundary condition: 
\[
V_X(s) = \frac{q(s)}{N}s + \left( 1 - \frac{q(s)}{N}\right)V_X({s+{1}/{N}}), \quad V_X(2)=1-\beta.
\]
The recursion can be written equivalently in terms of the function $s\mapsto W(s)=V_X(2-s)$, to get 
$W(s)=W(s-1/N)+f(s,W(s-1/N))/N$, where $f$ is the function such that 
\[
(t,v)\mapsto f(t,v)=\frac{\beta(2-t-v)}{1-\beta^2-t}.
\]
Then, the recursion for $W$ gives an Euler approximation to the following linear differential equation with initial condition: 
\begin{equation}
u'(s)=\beta\frac{2-s-u(s)}{1-\beta^2-s},\quad u(0)=1-\beta.\label{eq:UBODE}
\end{equation}
The function $\eta(s)=1-\beta+\beta s/(1-\beta)$ satisfies \eqref{eq:UBODE}, and by the Picard-Lindel\"{o}f theorem (see, e.g., \cite{teschl2012ordinary}) it is the unique solution to the system in the interval $(0,1-\beta^2)$.
Furthermore, the Euler approximation given by $W$ converges uniformly to $\eta$ (see, e.g.,~\cite{hairer1993solving}) and thus we conclude that $V_X(s)=W(2-s)\to \eta(2-s)=h(s)$ for every $s\in (1+\beta^2,2)$ as $N\to \infty$.
\end{proof}
\begin{proof}[Proof of \Cref{prop:UBSTp}]
From \Cref{lem:UBstep2}, the limit reward obtained by the algorithm using the policy with threshold $t$, as $N\to \infty$, is equal to 
\begin{equation}
\beta(1-\beta^2)+\EE[h(t-Z)]=\beta(1-\beta^2)+\beta h(t-(1-\beta^2))+(1-\beta)h(t).\label{eq:UBlimitreward}
\end{equation}
From the definition of $h$ and by studying the breakpoints in \eqref{eq:UBlimitreward}, we have that, as $N\to \infty$, the maximum possible reward achievable by a single-threshold policy in $\STp$ is $1+\beta^2$.
Then, by \Cref{lem:UBstep1}, we conclude the competitive ratio of any policy in $\STp$ is upper bounded by 
${(1+\beta^2)}/{(2+\EE[Z])}={(1+\beta^2)}/{(2+\beta-\beta^3)},$
and this function is minimized in $(0,1)$ at $\beta^*\approx 0.2067$, with value $\approx 0.4744$.
\end{proof}

\subsection{An Improved Lower Bound using the Threshold Dynamics}\label{sec:ODEimproving}

In this section, we show how to improve the $0.4$ lower bound from \Cref{prop:ST+LB} by using the differential equation from \Cref{sec:ode-approach}.
We prove the following.
\begin{theorem}\label{prop:improvedLB}
In the common-base model, there is a policy in $\STp$ with competitive ratio at least~$0.41$.
\end{theorem}
We prove \Cref{prop:improvedLB} by combining the lower bound \eqref{LB:STp} with the structure of the solution to our differential equation \eqref{eq:differential-t}.
In what follows, let $V_X$ be the solution to the differential equation \eqref{eq:differential-t} associated with the instance $X_1,\ldots,X_n$, and where $F$ and $f$ are the distribution and density, respectively, of $\max_{i\in [n]}X_i$.
Recall that from \Cref{thm:v-ode}, the function $V_X$ is quasiconcave and it has a unique maximizer $t^*\in [0,1]$.
In particular, at $\bE/2$, the differential equation \eqref{eq:differential-t} gives $$V_X'(\bE/2)=(V_X(\bE/2)-\bE/2)\frac{f(\bE/2)}{F(\bE/2)},$$ and by \Cref{cor:sc-lower-bound} we have $V_X(\bE/2)\ge \bE/2$, thus, $V'_X(\bE/2)\ge 0$.
Since $V_X$ is quasiconcave, we conclude that $t^*\ge \bE/2$.
In what follows, let $\alpha\ge 0$ such that $t^*=(1+\alpha)\bE/2$.
In the following lemma, we give a lower bound on the value of $V_X(\bE)$.
\begin{lemma}\label{lem:valueofBE}
$V_X(\bE) \geq \frac{3}{16}\cdot \frac{(1-\alpha)^2}{1+\alpha}\bE$.
\end{lemma}

\begin{proof}
Let $t_* = (1-\alpha)\bE/2$, i.e., $t^*+t_*=\bE$. 
For every $t\in [0,1]$, let $V^{\star}(t)=\min(t,V(t))$; we omit the subindex $X$ throughout the proof for notational simplicity. Consider the modified instance where we have a sequence of deterministic random variables whose values grow continuously from zero to $t^*$.
Formally, we fix a value $\varepsilon>0$, and for $i\in \{1,\dots,n-1\}$, we set $X_i = {t^*i}/{(n-1)}$
with probability $1$, and let $X_n$ be $t^*/\varepsilon$ with probability $\varepsilon$, and zero with probability $1 - \eps$.
As $\varepsilon\to 0$, the reward function for this new instance is $V^*$. Let $X^*_i$ denote the random variables in the new instance we just created, and let $\bar{X}_i=X^*_i-t^*$. 
We use $\bar{M}$ to denote $\max_{i\in [n]}\bar{X}_i$, and $M^{\star}$ for $\max_{i\in [n]}X^{\star}_i$. 
Observe that, by construction, we have $\bar\bE\ge t_*$.
Furthermore, $\bE^*\ge \bE$, and $M^*$ is stochastically dominated by $\max_{i\in [n]}X_i$. Then, we have the following equality:
\begin{equation}\label{eq:startobar}
V^{\star}(t^* + t)=\bar{V}(t)+t^*\PP[\bar M\ge t]\text{ for every }t\in [0,1-t_*].
\end{equation}
From \Cref{cor:sc-lower-bound}, it holds that $\bar{V}(t)\ge \min(t,\bar \bE-t)$ for every $t\in [0,1-t^*]$.
We also recall that $1-t^*\ge t_*$.
In particular, for every $t \in [0,{t_*}/{2}]$ we have $$\bar\bE-t\ge \bar\bE-t_*/2\ge t_*-t_*/2=t_*/2>t,$$ and therefore $\bar{V}(t)\ge t$ for every $t \in [0,{t_*}/{2}]$.
Thus, from \eqref{eq:startobar}, we have
\begin{equation}\label{eq:startobar2}
V^{\star}(t^*+t) \geq t + t^* \mathbb{P}[\bar{M} \geq t] \text{ for every $t \in [0,{t_*}/{2}]$.}
\end{equation}
By the quasiconcavity of $V^{\star}$, the function $t\mapsto V^{\star}(t^* + t)$ is non-increasing in the interval $[0,1-t^*]$ and therefore $V^{\star}(t^* + t)\le V^{\star}(t^*) \leq V(t^*)=t^*$, where the last equality holds from \Cref{thm:v-ode}.
Then, this inequality together with \eqref{eq:startobar2} implies that $t^* \geq t + t^* \mathbb{P}[\bar{M} \geq t]$ for every $t \in [0,{t_*}/{2}]$, i.e., 
\begin{equation}
\mathbb{P}[\bar M \geq t] \leq \frac{t^* - t}{t^*} = 1-\frac{t}{t^*}\text{ for every $t \in [0,{t_*}/{2}]$.}\label{eq:startobar3}
\end{equation}
Now consider $t \in [{t_*}/{2},t_*]$.
Since $t\mapsto \PP[\bar M \geq t]$ is decreasing in $t$, we have
\begin{equation}\label{eq:startobar4}
\PP[\bar M \geq t] \leq \PP[\bar M \geq t_*/2]\le 1-\frac{t_*}{2t^*}.
\end{equation}
Therefore, the following holds:
\begin{align*}
\int_0^{\bE}\PP[M^*\geq t] \dif t&= \int_0^{t^*}\PP[M^* \geq t]\dif t + \int_{t^*}^{\bE}\PP[M^* \geq t]\dif t \\
&\le t^* + \int_0^{t_*}\PP[\bar M \geq t]\dif t \\
&= t^* + \int_0^{t_*/2}\PP[\bar M \geq t]\dif t + \int_{t_*/2}^{t_*}\PP[\bar M \geq t]\dif t \\
&\le t^* + \int_0^{t_*/2} \left(1-\frac{t}{t^*}\right)\dif t + \int_{t_*/2}^{t_*}\left(1-\frac{t_*}{2t^*}\right)\dif t= \bE - \frac{3t_*^2}{8t^*},
\end{align*}
where the last inequality holds from \eqref{eq:startobar3} and \eqref{eq:startobar4}.
On the other hand, we have
\begin{align*}
\bE \le \bE^* &= \int_0^{\bE} \PP[M^* \geq t]\dif t +\int_{\bE}^1 \PP[M^* \geq t]\dif t\\
&\le 
\bE - \frac{3t_*^2}{8t^*} + \int_{\bE}^1 \PP[M^* \geq t]\dif t \le 
\bE - \frac{3t_*^2}{8t^*} + \int_{\bE}^1 (1-F(t))\dif t, 
\end{align*}
where the last inequality holds since $M^*$ is stochastically dominated by $\max_{i\in [n]}X_i$. Therefore, we have the following inequality:
\begin{equation}
\int_{\bE}^\infty (1-F(t))\dif t \geq \frac{3t_*^2}{8t^*}=\frac{3}{16}\cdot \frac{(1-\alpha)^2}{1+\alpha}\bE.\label{eq:startobar5}
\end{equation}
To conclude, recall that from \Cref{thm:v-ode}, for every $s\in [0,1]$ it holds
$$V(s) = s - F(s) + \int_s^1 \frac{F(s)}{F(u)} \dif u=s(1-F(s))+F(s)\int_{s}^1\frac{1-F(u)}{F(u)}\dif u,$$
and therefore, at $s=\bE$ we have  
\begin{align}
V(\bE)&=\bE(1-F(\bE))+F(\bE)\int_{\bE}^1\frac{1-F(u)}{F(u)}\dif u\nonumber\\
&\ge \bE(1-F(\bE))+F(\bE)\int_{\bE}^1(1-F(u))\dif u\nonumber\\
&\ge \left(\int_{\bE}^1(1-F(u))\dif u\right)(1-F(\bE))+F(\bE)\int_{\bE}^1(1-F(u))\dif u=\int_{\bE}^1(1-F(u))\dif u,\label{eq:startobar6}
\end{align}
where the first inequality holds from $F(u)\le 1$ for every $u\in [\bE,1]$, and the last inequality is a consequence of $\bE=\int_{0}^1(1-F(u))\dif u\ge \int_{\bE}^1(1-F(u))\dif u$. 
The lemma follows from \eqref{eq:startobar5}-\eqref{eq:startobar6}.
\end{proof}

\begin{proof}[Proof of Theorem \ref{prop:improvedLB}]
We consider a slight modification of the policy we use in \Cref{prop:ST+LB}.
Consider a random threshold $T$ such that it is equal to $T_1=(1+\delta)\bE/2$, with probability $0.8$, and it is equal to $\bE$ with probability $0.2$; $\delta>0$ is a parameter to be fixed later.
The bound \eqref{LB:STp} in particular holds when $Z$ is deterministic and equal to $z$, i.e., in the adversarial common-base model.
Then, for fixed $Z=z$, we have
$\EE[\hat V_Y(T)]=z+0.8V_X(T_1-z)+0.2V_X(\bE-z).$
We now analyze two different regimes for $\delta$.
First, consider $\delta\le \alpha$, where $\alpha$ is defined as in \Cref{lem:valueofBE}.
Since $\delta\le \alpha$, we have $T_1=(1+\delta)\bE/2\le t^*$, and therefore, \Cref{thm:v-ode} implies that 
\begin{equation}
V_X(T_1-z)=V_X((1+\delta)\bE/2-z)\ge \max(0,(1+\delta)\bE/2-z)=\max(0,T_1-z).\label{eq:betterUB1}
\end{equation}
Then, \eqref{eq:betterUB1} together with \eqref{LB:STp}, implies that
\begin{align*}
\EE[\hat V_Y(T)]&\ge z+0.8\max(0,(1+\delta)\bE/2-z)+0.2\max(0,\min(\bE-z,z))\\
&\ge z+0.8(T_1-z)\mathbf{1}_{z\le T_1}+0.2(\bE-z)\mathbf{1}_{\bE/2\le z\le \bE}+0.2z\mathbf{1}_{0\le z\le \bE/2}\\
&= z+0.8(T_1-z)\mathbf{1}_{z\le T_1}+0.2(\bE-z)\mathbf{1}_{\bE/2\le z\le T_1}+0.2(\bE-z)\mathbf{1}_{T_1\le z\le \bE}+0.2z\mathbf{1}_{0\le z\le \bE/2}\\
        &= (0.4\cdot z + 0.8\cdot T_1)\mathbf{1}_{z\leq {\bE}/{2}}+(0.8\cdot T_1+0.2\cdot\bE)\mathbf{1}_{{\bE}/{2} < z \leq T_1} +(0.8\cdot z+0.2\cdot \bE)\mathbf{1}_{T_1\le z}\\
        &\geq 0.8\cdot T_1 = 0.4(1+\delta)\bE,
\end{align*}
where the last inequality holds since each summand is at least $0.8\cdot T_1$ in the corresponding region.

Now consider $\delta\le 0.2\cdot \gamma(\alpha)$, where $\gamma(\alpha)=(3/16)(1-\alpha)^2/(1+\alpha)$.
Observe that if $z\le \bE/2$, the quasiconcavity of $V_X$ implies that $V_X(\bE-z)\ge \min(V_X(\bE/2),V_X(\bE))\ge \gamma(\alpha)\bE$, where the last inequality holds from \Cref{lem:valueofBE}.
Then, from \eqref{LB:STp}, we have the following inequalities: 
\begin{align*}
V_X(T_1-z)&\ge \max(0,\min(T_1-z,\bE-T_1+z))\\
&=\min(T_1-z,\bE-T_1+z)\mathbf{1}_{T_1-\bE\le z\le T_1}\\
&=((1-\delta){\bE}/{2}+z)\mathbf{1}_{z\le \delta {\bE}/{2}}+((1+\delta){\bE}/{2}-z)\mathbf{1}_{\delta {\bE}/{2}\le z\le T_1},\\
V_X(\bE-z)&\ge \max(z,\gamma(\alpha)\bE)\mathbf{1}_{z\le \bE/2}+\max(0,\bE-z)\mathbf{1}_{z\ge \bE/2}.
\end{align*}
Therefore, overall, we get
\begin{align*}
\EE[\hat V_Y(T)]&\ge z+0.8((1-\delta){\bE}/{2}+z)\mathbf{1}_{z\le \delta {\bE}/{2}}+0.8((1+\delta){\bE}/{2}-z)\mathbf{1}_{\delta {\bE}/{2}\le z\le T_1}\\
&\quad\quad+0.2\max(z,\gamma(\alpha)\bE)\mathbf{1}_{z\le \bE/2}+0.2\max(0,\bE-z)\mathbf{1}_{z\ge \bE/2}\\
&\geq (1.8\cdot z+(0.4-0.4\cdot \delta  + 0.2\cdot\gamma(\alpha))\bE)\mathbf{1}_{z\leq \delta {\bE}/{2}}+(0.4\cdot z+0.4(1+\delta)\bE)\mathbf{1}_{\delta {\bE}/{2} \le z \le {\bE}/{2}}\\
&\quad\quad+ (0.6+0.4\cdot\delta)\bE\cdot \mathbf{1}_{{\bE}/{2}\le z\le T_1}+(0.8\cdot z+0.2\cdot\bE)\mathbf{1}_{T_1\le z}\\
&\geq \min(0.4+0.2\cdot \gamma(\alpha)-0.4\cdot \delta,0.4+0.6\cdot \delta,0.6+0.4\cdot \delta)\bE\\
&\geq (0.4+0.6\cdot \delta)\bE\geq 0.4(1+\delta)\bE,
\end{align*}
where the fourth inequality holds since $\delta\le 0.2\cdot \gamma(\alpha)$.
By taking $\delta=\min_{\alpha\in [0,1]} \max(\alpha,0.2\cdot \gamma(\alpha)) \approx 0.033$, one of the two regimes always holds, yielding a lower bound of $0.4\cdot 1.033\cdot \bE > 0.41\cdot \bE$.
\end{proof}

\section{The Common-Scale Model}\label{sec:upper_bounds}

In this section, we pivot from an additive form of correlation to a multiplicative one. 
The correlations of the resulting model, which we refer to as the \emph{common-scale} model, prove to be significantly harder than those in the additive case.
In the following, we show that, in the common-scale model, neither the prophet nor the secretary objectives admit a non-trivial competitive ratio.
\begin{theorem}\label{thm:prophet_no_const_factor}
For every $n \in \mathbb{N}$ and every $\eps > 0$, there exist independent random variables $X_1,\dots,X_n$ and $Z$, with finite support, such that for $Y_i = ZX_i$ and all stopping times $\tau$, we have
\[
\frac{\E[Y_{\tau}]}{\E[\max_{i\in [n]}Y_i]} \leq \frac{1}{n} + \eps.
\]
\end{theorem}

\begin{theorem}\label{thm:secretary_no_const_factor}
For every $n \in \N$ and every $\mu > 0$ there exist independent random variables $X_1,\dots,X_n$ and $Z$, with finite support, such that for $Y_i = ZX_i$ and all stopping times $\tau$, we have
\[
\Pr\brk{Y_{\tau} = \max_{i\in [n]} Y_i} = O\left(\frac{1}{n^{1-\mu}}\right).
\]
\end{theorem}

We remark that the upper bounds in Theorems~\ref{thm:prophet_no_const_factor} and \ref{thm:secretary_no_const_factor} are essentially tight. For the prophet objective, the constant stopping rule that stops at time $i$, for which $\E[ZX_i]$ is maximal, achieves a competitive ratio of $1/n$. Similarly, by stopping at the time $i$ maximizing $\Pr\brk{X_i = \max_{j\in [n]} X_j}$, we get a competitive ratio of $1/n$ for the secretary objective.

We briefly describe our approach to proving Theorems~\ref{thm:prophet_no_const_factor} and~\ref{thm:secretary_no_const_factor}. Note that a simple subclass of stopping times $\tau$ consists of those depending only on the quotients $X_i/X_1$, since $Y_i/Y_1 = X_i/X_1$. We call such stopping times \emph{multiplicative invariant}.  In \Cref{sec:mult-invariance}, we prove that, given fixed random variables $X_1,\dots,X_n$, we can find a $Z$ such that multiplicative invariant stopping times perform almost as well as general stopping times. In the second step, we construct random variables $X_1,\dots,X_n$ for which multiplicative invariant stopping times perform badly. Combining these results directly leads to the proofs of Theorems~\ref{thm:prophet_no_const_factor} and~\ref{thm:secretary_no_const_factor}, in Sections~\ref{sec:hard_prophet} and~\ref{sec:hard_secretary}, respectively.
For the rest of this section, we assume that the supports of $Z$ and $X_1,\ldots,X_n$ are finite and that there exists an $L > 1$ such that all supports are contained in $L^{\mathbb Z} \coloneqq \set{L^z \midd z \in \ZZ}.$

\subsection{Multiplicative Invariant Stopping Times}\label{sec:mult-invariance}

Given random variables $Z$ and $X_1,\dots,X_n$, let $Y_i = ZX_i$ for every $i\in [n]$ and $\calT$ denote the set of all stopping times $\tau$ on the $Y_i$'s. Given $\tau \in \calT$, for every $k \in [n]$ and $y_1,\dots,y_k \in \mathbb Z$, whenever $(L^{y_1},\dots,L^{y_k})$ is in the support of $(Y_1,\dots,Y_k)$, we define
\begin{equation}\label{eq:delta-1}
\delta(y_1,\dots,y_k) = \Pr\brk{\tau = k \midd Y_1=L^{y_1},\dots,Y_k=L^{y_k}}.
\end{equation}
If $(L^{y_1},\dots,L^{y_k})$ is not in the support of $(Y_1,\dots,Y_k)$, we let $\delta(y_1,\dots,y_k)$ have an arbitrary value in $[0,1]$, subject to
\begin{equation}\label{eq:delta-2}
\textstyle\sum_{k=1}^n \delta(y_1,\dots,y_k) \leq 1, \quad \text{ for all } y_1,\dots,y_n \in \mathbb Z.
\end{equation}
Observe that, in the latter case, we can always choose $\delta(y_1,\dots,y_k) = 0$ to meet condition \eqref{eq:delta-2}. If $\tau$ and $\delta $ satisfy \eqref{eq:delta-1} and \eqref{eq:delta-2}, we say that $\tau$ induces $\delta$. Conversely, given an arbitrary $\delta$ assigning non-negative values to all tuples $(y_1,\dots,y_k)$ and satisfying \eqref{eq:delta-1} and \eqref{eq:delta-2}, there exists a stopping time $\tau$ inducing $\delta$. To see this, suppose we define $\tau$ in the following way: upon observing $Y_1=L^{y_1},\dots,Y_k=L^{y_k}$, we stop with probability
\[
\frac{\delta(y_1,\dots,y_k)}{1-\sum_{j=1}^{k-1}\delta(y_1,\dots,y_j)},
\]
and observe that \eqref{eq:delta-1} and \eqref{eq:delta-2} are satisfied by construction.
We are now ready to formally define multiplicative invariant stopping times.
\begin{definition}\label{def:mult_inv}
A stopping time $\tau^{\star} \in \calT$ is \emph{multiplicative invariant}, if it induces a $\delta^{\star}$, such that $\delta^{\star}(y_1,\dots,y_k) = \delta^{\star}(y_1+1,\dots,y_k+1)$ for every $k\in [n]$ and every $y_1,\dots,y_k \in \mathbb Z$. We denote the set of all multiplicative invariant stopping times by $\calT^{\star}$.
\end{definition}
Notice that \Cref{def:mult_inv} is equivalent to requiring that $\delta^{\star}(y_1+z,\dots,y_k+z) = \delta^{\star}(y_1,\dots,y_k)$ for all $z \in \ZZ$. In other words, multiplicative invariant stopping times depend only on the quotients $L^{y_i}/L^{y_1}$, or, equivalently, on the differences $y_i-y_1$; if $y_i-y_1 = y_i'-y_1'$ for all $i\in [k]$, then $y_i = y_i'-(y_1'-y_1)$, and for $\tau^{\star}$ multiplicative invariant,
\[
\delta^{\star}(y_1,\dots,y_k) = \delta^{\star}( y'_1-(y'_1-y_1),\dots,y_k'-(y_1'-y_1)) = \delta^{\star}(y_1',\dots,y_k').
\]
Conversely, if $\delta^{\star}(y_1,\dots,y_k) = \delta^{\star}(y_1',\dots,y_k')$ whenever $y_i-y_1 = y_i'-y_1'$, we can choose $y_i' = y_i+z$ to get
$\delta^{\star}(y_1+z,\dots,y_k+z) = \delta^{\star}(y_1',\dots,y_k') = \delta^{\star}(y_1,\dots,y_k).$
Importantly, since $\tau^{\star}$ is determined by the quotients $Y_i/Y_1 = X_i/X_1$, and the $X_i$ are independent of $Z$, $\tau^{\star}$ is also independent of $Z$, and so is $X_{\tau^{\star}}$.

Our first goal is to show that, for both the prophet and the secretary objectives, one can restrict their attention to multiplicative invariant stopping times.
\begin{lemma}\label{lem:mult-invariant}
Let $X_1,\ldots,X_n$ be random variables with bounded supports in $L^{\ZZ}$.
Then, for every $\eps>0$, there exists a corresponding random variable $Z$ with bounded support in $L^{\ZZ}$, such that for $Y_i = ZX_i$,
\begin{equation}\label{eq:prophet_mult_invariant}
\sup_{\tau \in \mathcal T} \frac{\E[Y_{\tau}]}{\E\brk{\max_{i\in [n]} Y_i}} \leq \sup_{\tau^{\star} \in \mathcal T^{\star}}\frac{\E[Y_{\tau^{\star}}]}{\E\brk{\max_{i\in [n]} Y_i}} + \eps,
\end{equation}
and
\begin{equation}\label{eq:secretary_mult_invariant}
\sup_{\tau \in \mathcal T} \Pr\brk{Y_{\tau} = \textstyle \max_{i\in [n]} Y_i} \leq \sup_{\tau^{\star} \in \mathcal T^{\star}} \Pr\brk{Y_{\tau^{\star}} = \textstyle \max_{i\in [n]} Y_i} + \eps.
\end{equation}
\end{lemma}
Before we prove \Cref{lem:mult-invariant}, we provide some intuition as to why it is true. Consider the prophet setting, where we seek to maximize $\E[ZX_{\tau}]$, and suppose that we can choose the distribution of $Z$ such that, in order for $\tau$ to perform well, $\tau$ must perform well for every realization of $Z$. If that is the case, then given $\tau$, we can construct a new, multiplicative invariant stopping time $\tau^{\star}$ that also performs well, in the following way: in the beginning of the game, draw a $Z' \in L^{\ZZ}$ uniformly at random from a bounded subset of $L^{\ZZ}$. Then, at each step, upon observing $Y_i$, present $Z' Y_i = (Z' Z) X_i$ to $\tau$, and stop whenever $\tau$ stops. Now by our assumption, $\tau$ would also perform well when the realization of $Z$ was the same as that of $Z'$. Hence, $\tau^{\star}$ should perform as well as $\tau$. Notice that, in some sense, we have decreased the dependence of $\tau^{\star}$ on $Z$. It turns out that, when choosing the support of $Z$ very large compared to the support of the $X_i$, this dependence vanishes in the limit and $\tau^{\star}$ becomes multiplicative invariant.

\Cref{lem:mult-invariant} follows from the more general \Cref{prop:mult_inv} that we present next. Let $\alpha: L^{\mathbb{Z}} \rightarrow \mathbb R_{> 0}$ be an arbitrary function, and $\beta: L^{\mathbb{Z}} \times \prn{L^{\mathbb{Z}}}^n \rightarrow \mathbb R_{> 0}$ be a bounded function. For $m \in \mathbb N$, let $\calD_m$ be the probability distribution of $Z$ where
\[
\Pr\brk{Z = L^z} = \frac{1}{\sum_{w = -m}^m {\alpha(L^z)}/{\alpha(L^w)}}, \quad \text{ for all } z \in \set{-m, \dots, m}.
\]
Observe that $\calD_m$ is a probability distribution for any $\alpha$ and $m$ since, for $\lambda_m = \sum_{w = -m}^m {1}/{\alpha(L^w)}$, we have
\[
\sum_{z = -m}^m \Pr\brk{Z = L^z} = \sum_{z = -m}^m \left(\sum_{w = -m}^m \frac{\alpha(L^z)}{\alpha(L^w)}\right)^{-1} = \sum_{z = -m}^m \frac{1}{\lambda_m \alpha(L^z)} = \frac{\lambda_m}{\lambda_m} = 1.
\]
\begin{proposition}\label{prop:mult_inv}
For all $\alpha: L^{\mathbb{Z}} \rightarrow \mathbb R_{> 0}, \beta: L^{\mathbb{Z}} \times \prn{L^{\mathbb{Z}}}^n \rightarrow \mathbb R_{> 0}$ and $m \in \NN$, let $Z$ be a random variable distributed according to $\calD_m$.
Then, for every $n\in \NN$ and independent random variables $X_1,\ldots,X_n$ with bounded support,\footnote{While we index $X_i$ with $\tau$ here instead of $Y_i=ZX_i$, $\{\tau=k\}$ is still determined by the variables $Y_1,\dots,Y_k$.} we have
\[
\sup_{\tau \in \mathcal T} \frac{\E[\alpha(Z)\beta(X_{\tau};X_1,\dots,X_n)]}{\E[\alpha(Z)]} \leq \sup_{\tau^{\star} \in \mathcal T^{\star}} \E[\beta(X_{\tau^{\star}};X_1,\dots,X_n)] + \frac{C(X,n,\beta)}{m}.
\]
$C(X,n,\beta)$ depends on the random variables $X_1,\dots,X_n$, $n$, and the function $\beta$, but not on $\alpha$ or $m$.
\end{proposition}

Intuitively, we will use $\alpha$ to relate $Z$ and $X_1, \dots, X_n$ and $\beta$ to denote the benchmark of a stopping time, i.e., the stopping time's competitive ratio or the probability that the stopping time equals the maximum realization.
We present the proof of the proposition, and then show how to prove \Cref{lem:mult-invariant}.
\begin{proof}[Proof of \Cref{prop:mult_inv}]
Let $\tau$ and $m$ be arbitrary. 
We will construct a multiplicative invariant stopping time $\tau^{\star}$ performing almost as well as $\tau$ in maximizing $U(\tau,m)$, where
$$U(\tau,m)=\frac{\E[\alpha(Z)\beta(X_{\tau};X_1,\dots,X_n)]}{\E[\alpha(Z)]}.$$
Let $\delta$ be induced by $\tau$ as described in the beginning of \Cref{sec:mult-invariance}.
For arbitrary $x_1,\dots,x_k \in \mathbb Z$, let
\begin{align*}
\bar \delta(x_1,\dots,x_k) &= \frac{1}{2m+1}\sum_{z=-m}^m \delta(z+x_1,\dots,z+x_k),\\
S(x_1,\dots,x_k) &= \Big\{z\in\mathbb{Z}\,:\, L^{z}L^{x_i} \in \supp(X_i) \text{ for every } i\in [k]\Big\}.
\end{align*}
We define a multiplicative invariant stopping time by setting
\[
\delta^{\star}(x_1,\dots,x_k) = 
\begin{cases}
    \min_{z \in S(x_1,\dots,x_k)} \bar \delta(z+x_1,\dots,z+x_k) &\text{ if } S(x_1,\dots,x_k) \neq \emptyset,\\
    0 &\text{ otherwise.}
\end{cases}
\]
Let us check that $\delta^{\star}$ is indeed induced by a multiplicative invariant stopping time. Observe that $z \in S(x_1,\dots,x_n) \Leftrightarrow z-1 \in S(x_1+1,\dots,x_n+1)$. Thus, when $S(x_1,\dots,x_k) \neq \emptyset$,
\begin{align*}
    \delta^{\star}(x_1,\dots,x_k) &= \min_{z \in S(x_1,\dots,x_k)} \bar \delta(z+x_1,\dots,z+x_k)\\&= \min_{z-1 \in S(x_1+1,\dots,x_k+1)} \bar \delta(z+x_1,\dots,z+x_k) \\&= \min_{z \in S(x_1+1,\dots,x_k+1)} \bar \delta(z+1+x_1,\dots,z+1+x_k) = \delta^{\star}(x_1+1,\dots,x_k+1).
\end{align*}
It is also easy to check that $\sum_{k=1}^n \delta^{\star}(x_1,\dots,x_k) \leq 1$ for all $x_1,\dots,x_n \in \mathbb Z$.
Thus, $\delta^{\star}$ is induced by a multiplicative invariant stopping time $\tau^{\star}$.

For any $L^{x_1},\dots,L^{x_k}$ in the support of $(X_1,\dots,X_k)$, there is $z \in S(x_1,\dots,x_k)$ such that
$\delta^{\star}(x_1,\dots,x_k) = \bar \delta(x_1+z,\dots,x_k+z),$
so we have 
\begin{align*}
    \delta^{\star}(x_1,\dots,x_k) &= \frac{1}{2m+1}\sum_{z'=-m}^m \delta(x_1+z+z',\dots,x_k+z+z')\\
    &=\frac{1}{2m+1}\sum_{z'=z-m}^{z+m} \delta(x_1+z',\dots,x_k+z').
\end{align*}
Consequently, we can write $\delta^{\star}(x_1,\dots,x_k)-\bar \delta(x_1,\dots,x_k)$ as
\begin{align*}
\frac{1}{2m+1}\left(\sum_{z'=z-m}^{z+m} \delta(x_1+z',\dots,x_k+z')-\sum_{z'=-m}^m\delta(x_1+z',\dots,x_k+z')\right).
\end{align*}
There are exactly $2|z|$ indices $z'$ appearing in one sum but not in the other one, so we can bound
\[
|\delta^{\star}(x_1,\dots,x_k)-\bar \delta(x_1,\dots,x_k)| \leq \frac{2|z|}{2m+1} \leq \frac{c}{m}, \;\text{ for every } (L^{x_1},\dots,L^{x_k}) \in \supp (X_1,\dots,X_k),
\]
where $c = \max_{x_1 \in \supp X_1} \max_{z \in S(x_1)} |z|.$
Note that $c < \infty$ because $\supp X_1$ is finite.

For $\mathbf x \in \mathbb Z^n$, write $f(\mathbf x) = \prod_{i=1}^n \Pr\brk{X_i = L^{x_i}}$ and $L^{\mathbf x} = (L^{x_1},\dots,L^{x_n})$. Then we can write
\begin{align*}
\Pr\brk{X_i = L^{x_i} \text{ for all } i \in [n] \wedge \tau^{\star}} = k) &= \Pr\brk{\tau^{\star} = k \mid X_i = L^{x_i}, \, \forall i \in [n]} \cdot \Pr\brk{X_i = L^{x_i}, \, \forall i \in [n]} \\
&= \Pr\brk{\tau^{\star} = k \mid X_i=L^{x_i}, \, \forall i \in [k]} \cdot f(\mathbf x) \\
&= \delta^{\star}(x_1,\dots,x_k)f(\mathbf x),
\end{align*}
where we used that the event $\{\tau = k\}$ only depends on the random variables $X_1,\dots,X_k$ in the second-to-last equation. So we can write
\begin{align*}
\E[\beta(X_{\tau^{\star}};X_1,\dots,X_n)] &= \sum_{\mathbf x \in \mathbb Z^n}\sum_{k=1}^n \beta(L^{x_k}; L^{\mathbf x}) \Pr\brk{X_i = L^{x_i} \, \forall i \in [n] \wedge \tau = k} \\
&= \sum_{\mathbf x \in \mathbb Z^n}\sum_{k=1}^n \beta(L^{x_k};L^{\mathbf x}) \delta^{\star}(x_1,\dots,x_k)f(\mathbf x) \\ 
&\geq \sum_{\mathbf x \in \mathbb Z^n}\sum_{k=1}^n \beta(L^{x_k};L^{\mathbf x})\bar \delta(x_1,\dots,x_k)f(\mathbf x) - \frac{c}{m}\sum_{\mathbf x \in \mathbb Z^n}\sum_{k=1}^n \beta(L^{x_k};L^{\mathbf x})f(\mathbf x) \\ 
&\geq \frac{1}{2m+1}\sum_{z=-m}^m \sum_{\mathbf x \in \mathbb Z^n}\sum_{k=1}^n \beta(L^{x_k};L^{\mathbf x}) \delta(x_1+z,\dots,x_k+z)f(\mathbf x) - \frac{cn\lVert \beta \rVert_{\infty}}{m}.
\end{align*}
Let $\lambda_m = \sum_{z=-m}^m \frac{1}{\alpha(L^z)}$. Since $\Pr\brk{Z=L^z} \alpha(L^z) = \frac{1}{\lambda_m}$ for all $z \in \{-m,\dots,m\}$, we have
\begin{align*}
&\frac{1}{2m+1}\sum_{z=-m}^m \sum_{\mathbf x \in \mathbb Z^n}\sum_{k=1}^n \beta(x_k;\mathbf x) \delta(x_1+z,\dots,x_k+z)f(\mathbf x) \\ 
&= \frac{\lambda_m}{2m+1}\sum_{z=-m}^m \sum_{\mathbf x \in \mathbb Z^n}\sum_{k=1}^n \alpha(L^z)\beta(L^{x_k};L^{\mathbf x}) \delta(x_1+z,\dots,x_k+z)f(\mathbf x)\Pr\brk{Z = L^z} \\
&= \frac{\lambda_m \E[\alpha(Z)\beta(X_{\tau};X_1,\dots,X_n)]}{2m+1}.
\end{align*}
Finally, we note that $\E[\alpha(Z)] = \frac{2m+1}{\lambda_m}$ to arrive at
\[
\frac{\E[\alpha(Z)\beta(X_{\tau};X_1,\dots,X_n)]}{\E[\alpha(Z)]} \leq \E[\beta(X_{\tau^{\star}};X_1,\dots,X_n)] + \frac{cn\lVert \beta \rVert_{\infty}}{m}
\]
and the claim follows by choosing
$C(X,n,\beta) := cn\lVert \beta \rVert_{\infty}.$
\end{proof}
\begin{proof}[Proof of \Cref{lem:mult-invariant}]
For \eqref{eq:prophet_mult_invariant}, we let $\alpha(\cdot)$ be the identity function, i.e. $\alpha(L^z) = L^z$ and consider
$\beta(L^{\xi};L^{x_1},\dots,L^{x_n}) = {L^{\xi}}/{\E[\max_{i\in [n]} X_i]},$
for $(L^{x_1},\dots,L^{x_n})$ in the support of $(X_1,\dots,X_n)$, and zero otherwise. Since $X_1,\ldots,X_n$ have bounded support and $\beta$ is bounded, we can apply \Cref{prop:mult_inv} to deduce
\begin{align*}
\sup_{\tau \in \mathcal T} \frac{Y_{\tau}}{\E[\max_{i \in [n]} Y_i]} &= \sup_{\tau \in \mathcal T} \frac{\E[Z X_{\tau}]}{\E[Z]\cdot \E[\max_{i \in [n]} X_i]} \\
&= \sup_{\tau \in \mathcal T} \frac{\E[\alpha(Z)\beta(X_{\tau};X_1,\dots,X_n)]}{\E[\alpha(Z)]} \\
&\leq \sup_{\tau^{\star} \in \mathcal T^{\star}} \E[\beta(X_{\tau^{\star}};X_1,\dots,X_n)] + \frac{C(X,n,\beta)}{m} \\
&= \sup_{\tau^{\star} \in \mathcal T^{\star}} \E[\beta(Y_{\tau^{\star}};Y_1,\dots,Y_n)] + \frac{C(X,n,\beta)}{m}\\
&= \sup_{\tau^{\star} \in \mathcal T^{\star}} \frac{\E[Y_{\tau^{\star}}]}{\E[\max_{i \in [n]} Y_i]} + \frac{C(X,n,\beta)}{m},
\end{align*}
where, for the second-to-last equality, we used the fact that
\[
\beta(X_{\tau^{\star}};X_1,\dots,X_n) = \frac{X_\tau^{\star}}{\E[\max_{i\in [n]} X_i]} = \frac{Z X_\tau^{\star}}{Z \E[\max_{i\in [n]} X_i]} \overset{(*)}{=} \frac{Y_\tau^{\star}}{\E[\max_{i\in [n]} Y_i]} = \beta(Y_{\tau^{\star}};Y_1,\dots,Y_n),
\]
where $(*)$ follows from the fact that, if $\tau^{\star}$ is multiplicative invariant, $X_{\tau^{\star}}$ is independent of the value of $Z$, and from the fact that $\argmax_{i \in [n]} Y_i$ is independent of the value of $Z$.
For \eqref{eq:secretary_mult_invariant}, we let $\alpha$ be the constant function equal to $1$, i.e., $\alpha(L^z) = 1$, and $\beta(L^{\xi};L^{x_1},\dots,L^{x_n}) = \mathbf{1}\brk{\xi = \max_{i\in [n]} x_i}$, where $\mathbf{1}\brk{\calE}$ denotes the indicator function of event $\calE$. Then, for $X_\tau = L^\xi$ and $X_{\tau^{\star}} = L^{\psi'}$ for some $\xi$ and $\psi$, for $(L^{x_1}, \dots, L^{x_n})$ in the support of $(X_1, \dots, X_n)$, we have
\begin{align*}
\sup_{\tau \in \mathcal T} \Pr\brk{Y_{\tau} = \textstyle \max_{i\in [n]} Y_i} &= \sup_{\tau \in \mathcal T} \Pr\brk{L^z L^\xi = \textstyle \max_{i\in [n]} L^z L^{x_i}} \\
&= \sup_{\tau \in \mathcal T} \Pr\brk{L^\xi = \textstyle \max_{i\in [n]} L^{x_i}} \\
&= \sup_{\tau \in \mathcal T} \frac{\E\brk{\alpha(Z) \beta(X_\tau;X_1,\dots,X_n)}}{\E[\alpha(Z)]} \\
&\leq \sup_{\tau^{\star} \in \calT^{\star}} \E\brk{\beta(X_{\tau^{\star}};X_1,\dots,X_n)} + \frac{C(X,n,\beta)}{m} \\
&= \sup_{\tau^{\star} \in \calT^{\star}} \Pr\brk{L^{\psi} = \textstyle \max_{i\in [n]} L^{x_i}} + \frac{C(X,n,\beta)}{m} \\
&= \sup_{\tau^{\star} \in \calT^{\star}} \Pr\brk{L^z L^{\psi} = \textstyle \max_{i\in [n]} L^z L^{x_i}} + \frac{C(X,n,\beta)}{m} \\
&= \Pr\brk{Y_{\tau^{\star}} = \textstyle \max_{i\in [n]} Y_i} + \frac{C(X,n,\beta)}{m}.
\end{align*}
In both cases, for every $\eps > 0$, there exists an $m$ large enough to satisfy \Cref{prop:mult_inv}.
\end{proof}

\subsection{Hard Instance for the Prophet Setting}\label{sec:hard_prophet}

In this section, we prove \Cref{thm:prophet_no_const_factor}. Given \Cref{lem:mult-invariant}, it remains to construct instances $X_1,\dots,X_n$, such that no multiplicative invariant stopping time achieves a competitive ratio significantly better than $1/n$. To give some intuition for the construction, assume that $X_i = L^{\xi_i}$ where, for some $\mu > 0$, with probability $\mu$, $\xi_i$ is $1+1/M$ for some large value $M$ and with probability $1-\mu$, $\xi_i$ is uniformly chosen from $I_i \subseteq [0,1]$. Now, no matter how large $M$ is and how small $\mu$ is, we can always make $L$ so large that the only significant contribution to $\E[X_{\tau}]$ comes from realizations of the random variables where one of the $\xi_i$'s is equal to $1+\eps$, and $\tau$ stops at this index $i$. We would like to make it difficult for an algorithm, which only sees the quotients $X_i/X_1$, to detect whether the current $X_i$ has $\xi_i$ equal to $1+1/M$, or not. To this end, set $I_1 = [0,1]$. Now, when observing $X_2/X_1$, we would guess that $\xi_2 = 1+1/M$, when $X_2/X_1$ is very close to $1$. However, when we set $I_2 = [1-\alpha,1]$ for a small $\alpha$, the ratio $X_2/X_1$ is always close to $1$. We continue like this by setting $I_3 = [1-\alpha^2,1]$, and in general, $I_i = [1-\alpha^{i-1},1]$. Then it remains hard to infer any non-trivial information from the ratios $X_i/X_{i-1}$, as they are always very close to $1$. In fact, the only informative case occurs when $X_i/X_{i-1} < 1$, telling us that with very high probability, $\xi_{i-1}$ was equal to $1+1/M$. Unfortunately, since the probability of seeing more than one variable attaining this value is of the order of $\mu^2$, the contributions of these cases are negligible.

In the following, we provide a formal construction, which is a discretized version of the instance described above. To that goal, let $n \in \mathbb{N}$ be arbitrary, $\mu \in (0,1) \cap \mathbb Q$ and $\alpha = \mu^2$. Let $M > 0$ be arbitrary such that $\alpha^iM \in \mathbb{N}$ for all $i\in \{1,\dots,n\}$. For $i\in \{1,\dots,n\}$, consider the random variables given by
\[
\xi_i = \begin{cases}j &\text{ w.p. } p_i=\frac{1-\mu}{M(1-\alpha^{i-1})} \text{ for } j \in J_i=\{(1-\alpha^{i-1})M,\dots,M-1\},\\M &\text{ w.p. } \mu,
\end{cases}
\]
and set $X_i = L^{\xi_i}$, where as before $L > 1$ is a constant. We are going to prove that for fixed $n$ and any $\eps > 0$ we can choose $L$ and $\mu$ such that
\[
\frac{\E[X_{\tau^{\star}}]}{\E[\max_{i\in [n]} X_i]} \leq \frac{1}{n} +\eps
\]
for all multiplicative invariant stopping times $\tau^{\star} \in \mathcal T^{\star}$.

As mentioned, $\tau^{\star}$ is multiplicative invariant if and only if it depends on the differences $\xi_i-\xi_1$, which in turn is equivalent to depending on the differences $D_i = \xi_i-\xi_{i-1}$ for $i \geq 2$. This viewpoint is a bit more practical, so in what follows, we will think of a multiplicative invariant stopping time as depending only on the observations $D_2,\dots,D_i$.
Define $u_i = M\alpha^{i-2}(1-\alpha)$ for $i\in \{2,\dots,n\}$, and denote by $\calE(d_2,\dots,d_k)$ the event in which $D_i = d_i$ for every $i\in \{2,\dots,k\}$. 
The core of \Cref{thm:prophet_no_const_factor} is in the following lemma. 
\begin{lemma}\label{lem:prophet_value_low}
For any stopping time $\tau^{\star} \in \mathcal T^{\star}$, we have
$\E[X_{\tau^\star}] \leq L^M(\mu + n(2\alpha + n^2\mu^2)) + L^{M-1}$.
\end{lemma}
To show the lemma, we first argue that, under the event that up to some time $i$ all observed differences $D_i$ lie between $1$ and $u_i$, these observed differences in fact provide \emph{no} information about $\xi_i$ and any possible subsequent observations.
Next, we prove that observations that would provide sufficient information occur very rarely.
Finally, we show that for any multiplicative invariant stopping time, the probability of stopping at the right point in time is roughly $\mu$, which is essentially the same as if we always stopped at, say, $X_1$. 
The full proof can be found in \Cref{app:lem:prophet_value_low}.
\begin{proof}[Proof of \Cref{thm:prophet_no_const_factor}]
Let $\mu \in (0,1)$ be an arbitrary rational number with $\mu < n^{-4}$. Choose $\alpha = \mu^2$, $L=\mu^{-2}$, and choose $M$ such that $M\alpha^i$ is a natural number for all $i\in \{1,\dots,n\}$. 
We have
\begin{align*}
\E[\textstyle\max_{i\in [n]} X_i] \geq L^M\Pr\brk{\exists i: X_i = L^M} = L^M(1-(1-\mu)^n).
\end{align*}
On the other hand,
\[
(1-\mu)^n = 1-n\mu + \frac{n(n-1)}{2}(1-\xi)^{n-1}\xi^2
\]
for some $\xi \in [0,\mu]$, so using that $\mu < n^{-4}$, we have 
$(1-\mu)^n \leq 1-n\mu + \mu^{3/2}.$
Therefore, we get
$\E[\textstyle\max_{i\in [n]} X_i] \geq L^M(n\mu - \mu^{3/2}).$
Combining this with \Cref{lem:prophet_value_low}, for every $\tau^{\star}\in \calT^{\star}$ we get
\begin{align*}
    \frac{\E[X_{\tau^{\star}}]}{\E[\max_{i\in [n]} X_i]} \leq \frac{1+ 2n\alpha/\mu + n^3\mu}{n - \mu^{1/2}} + \frac{1}{L(n\mu-\mu^{3/2})} = \frac{1+ 2n\mu + n^3\mu}{n - \mu^{1/2}} + \frac{\mu}{n-\mu^{1/2}}.
\end{align*}
Now, letting $\mu\to 0$ proves the result.
\end{proof}

\subsection{Hard Instance for the Secretary Setting}\label{sec:hard_secretary}

Next, we construct instances that are hard for the secretary objective, i.e., $\Pr\brk{X_{\tau^{\star}} = \max_{i \in [n]} X_i}$ approaches zero for every multiplicative invariant stopping time $\tau^{\star}$ as $n \rightarrow \infty$. Given \Cref{sec:hard_prophet}, we already know how to construct instances for which it is hard to stop at the random variable attaining the maximum value of its support. In the prophet setting, the probability that $X_i$ attains its maximum value was a very small number $\mu$. In the secretary setting, however, the $X_i$ must be assigned their maximum value with a probability $\kappa(n)$, such that very likely at least one of them actually attains it -- otherwise, it would be almost optimal to always accept $X_n$. We also let the maximal value of the $X_i$ be slightly decreasing as $i$ increases. By doing this, we force the algorithm to stop at the \emph{first} index $i$, such that $X_i$ attains its maximal value. The formal analysis is quite similar to that of \Cref{sec:hard_prophet}, however, it is slightly more involved.

Let $n$ be arbitrary and set $\alpha = n^{-2}$. Let $M$ be such that $\alpha^iM$ is an integer for all $i\in \{1,\dots,k\}$, and $n/(\alpha^nM) \leq \alpha$. Set $\kappa(n) = n^{-1+\mu}$ for some $\mu \in (0,1)$. Define
\[
\xi_i = \begin{cases}
\xi_i^{\star}=M + n  - i  &\text{ w.p. } \kappa(n), \\
j &\text{ w.p. } p_i=\frac{1-\kappa(n)}{\alpha^{i-1}M}\text{ for } j \in J_i=\{(1-\alpha^{i-1})M,\dots,M-1\},
\end{cases}
\]
and set $X_i = 2^{\xi_i}$. Note that by taking logarithms,
$\textstyle\Pr[{X_{\tau^{\star}} = \max_{i\in [n]} X_i}] = \Pr[{\xi_{\tau^{\star}} = \max_{i\in [n]} \xi_i}].$
Also recall that $\tau^{\star}$ being multiplicative invariant with respect to the $X_i$ is equivalent to $\tau^{\star}$ only depending on the differences $\xi_i-\xi_{i-1}$. Similar to the prophet setting, the idea of our construction is that with high probability, it is hard to infer from the size of the jump from $\xi_i$ to $\xi_{i+1}$ whether we arrived at the maximum value $\xi_{i+1}$.
As before, set $D_i = \xi_i - \xi_{i-1}$ for $i\in \{2,\dots,n\}$, $u_i = M\alpha^{i-2}(1-\alpha)$ and denote by $\calE(d_2,\dots,d_k)$ the event in which $D_i = d_i$ for every $i\in \{2,\dots,k\}$. 
Similar to the prophet setting, the following lemma allows us to show \Cref{thm:secretary_no_const_factor}, and its proof is in \Cref{app:lem:max_prob_bound}.
\begin{lemma}\label{lem:max_prob_bound}
    Let $\tau^{\star} \in \mathcal T^{\star}$ be any multiplicative invariant stopping time. Let $W$ be the event in which there exists $i\in [n]$ with $\xi_i = \xi_i^{\star}$. Then 
    $\textstyle\Pr\brk{\xi_{\tau^{\star}} = \max_{i\in [n]} \xi_i \wedge W\midd \calE(d_2,\dots,d_k)} \leq \kappa(n) + 4n\alpha.$
\end{lemma}
\begin{proof}[Proof of \Cref{thm:secretary_no_const_factor}]
By \Cref{lem:max_prob_bound} we have
\begin{align*}
\textstyle\Pr\brk{\xi_{\tau^{\star}} = \max_{i\in [n]} \xi_i} &\leq \textstyle\Pr\brk{\xi_{\tau^{\star}} = \max_{i\in [n]} \xi_i \wedge W} + \Pr[W^c] \\ 
&\leq \kappa(n) + 4(n-1)\alpha + (1-\kappa(n))^n\leq \kappa(n)+\frac{4}{n} + (1-\kappa(n))^n.
\end{align*}
Both $4/n$ and $(1-\kappa(n))^n$ are in $o({n^{-1+\mu}})$, which proves the theorem.
\end{proof}

\subsection{Proof of \Cref{lem:prophet_value_low}}\label{app:lem:prophet_value_low}
The following lemma states that, under the event that up to some time $i$ all observed differences $D_i$ lie between $1$ and $u_i$, these observed differences in fact provide \emph{no} information about $\xi_i$ and any possible subsequent observations.
\begin{lemma}\label{lem:prophet_independence}
    Let $1 \leq d_i \leq u_i$ for $i\in \{2,\dots,k\}$. 
    Then $\xi_k,\dots,\xi_n$ and $\calE(d_2,\dots,d_k)$ are independent.
\end{lemma}
\begin{proof}
    The events $\xi_k = y_k$ and $\calE(d_2,\dots,d_k)$ imply that $ \xi_i = y_i = y_k - \sum_{j=i+1}^k d_j$ for all $i\in \{1,\dots,k-1\}$. 
    Now since $d_i \geq 1$, clearly $\xi_i < M$ for $i < k$. Moreover,
    \begin{align*}
        y_i &\geq (1-\alpha^{k-1})M - M(1-\alpha)\sum_{j=i+1}^k\alpha^{j-2}\\
        &= (1-\alpha^{k-1})M - M(1-\alpha)\alpha^{i-1}\sum_{j=0}^{k-i-1}\alpha^{j}\\ 
        &= (1-\alpha^{k-1})M - M\alpha^{i-1}(1-\alpha^{k-i})= M(1-\alpha^{i-1}),
    \end{align*}
    so $y_i \in J_i$ for all $i\in \{1,\dots,k-1\}$. Consequently,
    \begin{align*}
    \Pr\brk{\xi_k=y_k \wedge \dots \wedge \xi_n = y_n \wedge \calE(d_2,\dots,d_k)} = \prod_{i=k}^n \Pr\brk{\xi_i = y_i} \prod_{i=1}^{k-1} p_i.
    \end{align*}
    This implies independence since
    $\Pr\brk{\calE(d_2,\dots,d_k)} = \prod_{i=1}^{k-1} p_i.$
\end{proof}
Next we prove that observations that would provide sufficient information, in particular those that are not covered by \Cref{lem:prophet_independence}, occur very rarely.
\begin{lemma}\label{lem:decreasing_unlikely}
$\Pr\brk{D_{k+1} \geq u_{k+1}} \leq \alpha,$
and 
$\Pr\brk{D_{k+1} \leq 0 \wedge \xi_k < M} \leq \alpha.$
\end{lemma}
\begin{proof}
$D_{k+1} \geq u_{k+1}$ implies
$\xi_k \leq M - u_{k+1} = M - M\alpha^{k-1}(1-\alpha) = M(1-\alpha^{k-1}+\alpha^k),$
but
\[
\Pr(\xi_k \leq M(1-\alpha^{k-1}+\alpha^k)) = M\alpha^kp_k = (1-\mu)\alpha,
\]
which proves the first inequality.
On the other hand, $D_{k+1} \leq 0$ and $\xi_k < M$ imply $M-1 \geq \xi_k \geq M(1-\alpha^k)$, but
$\Pr\brk{M-1 \geq \xi_k \geq M(1-\alpha^k)} = M\alpha^kp_k = \alpha$.
\end{proof}
The following lemma says that for any multiplicative invariant stopping time, the probability of stopping at the right point in time is roughly $\mu$, which is essentially the same as if we always stopped at, say, $X_1$. The proof works by inductively bounding the optimal online algorithm.
\begin{lemma}\label{lem:prophet_find_max_is_hard}
For any stopping time $\tau^{\star} \in \mathcal T^{\star}$, $\Pr\brk{\xi_{\tau^{\star}} = M} \leq \mu + n(2\alpha+n^2\mu^2)$.
\end{lemma}
\begin{proof}
We prove by induction that for all $k\in \{1,\dots,n\}$ and all $1 \leq d_i \leq u_i$, there holds
\[
\Pr\brk{\xi_{\tau^{\star} \lor k} = M \midd \calE(d_2,\dots,d_k)} \leq \mu + (n-k)(2\alpha+n^2\mu^2),
\]
which implies the result. Here, $\tau^{\star} \lor k = \max\{\tau^{\star},k\}$.
For $k=n$, we have that
\begin{align}
    \Pr\brk{\xi_{\tau^{\star} \lor n} = M \midd \calE(d_1,\dots,d_n)} \leq \Pr\brk{\xi_n = M \midd \calE(d_2,\dots,d_n)} = \mu,
\end{align}
by independence (\Cref{lem:prophet_independence}).
Let $k < n$ and suppose the claim is true for $k+1,\dots,n$. Let
$\nu = \Pr(\tau^{\star} = k\mid E(d_2,\dots,d_k)).$
Then
\begin{align*}
&\Pr\brk{\xi_{\tau^{\star} \lor k} =M \midd \calE(d_2,\dots,d_k)}\\
&= \nu \Pr\brk{\xi_k = M \midd \calE(d_2,\dots,d_k)} + (1-\nu) \Pr(\brk{\xi_{\tau^{\star} \lor k+1} = M \midd \calE(d_2,\dots,d_k)}.
\end{align*}
Again, by independence, the first term is bounded by $\mu^2$. Regarding the second term, write
\begin{align*}
    &\Pr\brk{\xi_{\tau^{\star} \lor k+1} = M \midd \calE(d_2,\dots,d_k)} \\ 
    &= \sum_{1\leq d_{k+1}\leq u_{k+1}} \Pr\brk{\xi_{\tau^{\star} \lor k+1} = M \midd \calE(d_2,\dots,d_k)} \Pr\brk{D_{k+1} = d_{k+1} \midd \calE(d_2,\dots,d_{k+1})} \\
    &\quad+ \Pr\brk{\xi_{\tau^{\star} \lor k+1} = M \wedge D_{k+1} > u_{k+1} \midd \calE(d_2,\dots,d_k)} \\
    &\quad+ \Pr\brk{\xi_{\tau^{\star} \lor k+1} = M \wedge D_{k+1} \leq 0 \wedge \xi_k < M \midd \calE(d_2,\dots,d_k)} \\
    &\quad+
    \Pr\brk{\xi_{\tau^{\star} \lor k+1} = M \wedge D_{k+1} \leq 0 \wedge \xi_k = M \midd \calE(d_2,\dots,d_k)}.
\end{align*}
By the induction hypothesis, we can bound the first term by $\mu + (2\alpha+\mu^2)(n-k-1)$. By \Cref{lem:decreasing_unlikely} and independence, we can bound the second and the third terms respectively by $\alpha$. Regarding the last term, note that
$\xi_{\tau^{\star} \lor k+1} = M \wedge \xi_k = M$
implies the event $F = \{\exists i > j \geq k:\, \xi_i=\xi_j=M\}$. By independence,
$\Pr\brk{F \midd \calE(d_2,\dots,d_k)} = \Pr[F] \leq n^2\mu^2.$
Using these upper bounds, we get
\begin{align*}
\Pr\brk{\xi_{\tau^{\star} \lor k+1} = M \midd \calE(d_2,\dots,d_k)}&\leq \mu + (2\alpha + n^2\mu^2)(n-k-1) + 2\alpha + n^2\mu^2  \\
&= \mu + (2\alpha + n^2\mu^2)(n-k),
\end{align*}
which concludes the proof of the lemma.
\end{proof}
\begin{proof}[Proof of \Cref{lem:prophet_value_low}]
We have
$\E[X_{\tau^\star}] \leq L^M \Pr\brk{\xi_{\tau^\star} = M} + L^{M-1} \Pr\brk{\xi_{\tau^\star}\leq M-1}$, and by \Cref{lem:prophet_find_max_is_hard}, it further holds $\Pr\brk{\xi_{\tau^\star} = M} \leq \mu + (2\alpha +n^2\mu^2)n$.
\end{proof}

\subsection{Proof of \Cref{lem:max_prob_bound}}\label{app:lem:max_prob_bound}
First, we state an independence result similar to \Cref{lem:prophet_independence}.
\begin{lemma}\label{lem:secretary_independence}
Let $n \leq d_i \leq u_i$ for $i\in \{2,\dots,k\}$. Then the random variables $\xi_k,\dots,\xi_n$ and the event $\calE(d_2,\dots,d_k)$ are independent.
\end{lemma}
\begin{proof}
    The events $\calE(d_2,\dots,d_k)$ and $\xi_k = y_k$ imply that $\xi_i = y_i := y_k - \sum_{j=i+1}^k d_j$ for $i=1,\dots,k-1$. Now since $d_i \geq n$, this implies that $y_i \leq M-1$ for $i < k$. We also have that
    \begin{align*}
    y_i &\geq (1-\alpha^{k-1})M - \sum_{j=i+1}^k u_j \\ &= (1-\alpha^{k-1})M - M(1-\alpha)\sum_{j=i+1}^k \alpha^{j-2} \\&= (1-\alpha^{k-1})M - M(1-\alpha)\alpha^{i-1}\sum_{j=0}^{k-i-1} \alpha^{j} \\ &=
    (1-\alpha^{k-1})M - M\alpha^{i-1}(1-\alpha^{k-i})= (1-\alpha^{i-1})M,
    \end{align*}
    so $y_i \in J_i$ and $\Pr\brk{\xi_i = y_i} = p_i$ for all $i=1,\dots,k-1$. Thus,
    \begin{align*}
    \Pr\brk{\xi_k=y_k \wedge \dots, \xi_n = y_n \wedge \calE(d_2,\dots,d_k)} = \Pr[\xi_k = y_k] \cdots \Pr[\xi_n=y_n] \prod_{i=1}^{k-1}p_i.
    \end{align*}
    Summing over all values $y_k,\dots,y_n$, we see that $\Pr[E] = \prod_{i=1}^{k-1} p_i$ which proves the result.
\end{proof}
Now we show that certain events, which would help recognize the maximal value $\xi_i^{\star}$, are unlikely to happen.
\begin{lemma}\label{lem:informative_events_unlikely}
The following holds:
\begin{enumerate}[itemsep=0pt,label=\normalfont(\roman*)]
    \item $\Pr\brk{\xi_k < \xi_k^{\star} \wedge D_{k+1} < 0} \leq \alpha$.\label{informative1}
    
    \item $\Pr[0 \leq D_{k+1} \leq n] \leq 2\alpha$.\label{informative2}
    
    \item $\Pr[D_{k+1} \geq u_{k+1}] \leq \alpha.$\label{informative3}
\end{enumerate}
\end{lemma}
\begin{proof}
    Observe that $D_{k+1} < 0$ implies that $\xi_k \geq \xi_{k+1} \geq (1-\alpha^k)M$. Now since also $\xi_k < \xi_k^{\star}$, we have that $\xi_k \in \{(1-\alpha^k)M,\dots,M-1\}$, so
    \[
    \Pr\brk{\xi_k < \xi_k^{\star} \wedge D_{k+1} < 0} \leq \alpha^kM p_i = \alpha(1-\kappa(n)) \leq \alpha.
    \]
    The event $0 \leq D_{k+1} \leq n$ implies that $(1-\alpha^k)M - n \leq \xi_k \leq M-1$, so
    \[
    \Pr\brk{0 \leq D_{k+1} \leq n} \leq ((1-\alpha^k)M + n)p_k \leq \alpha + \frac{n}{M\alpha^{k-1}} \leq 2\alpha.
    \]
    Finally, the event $D_{k+1} \geq u_{k+1}$ implies that
    \[
    \xi_{k} \leq M-1 - M\alpha^{k-1}(1-\alpha) \leq M(1-\alpha^{k-1} + \alpha^k).
    \]
    Thus,
    $\Pr[D_{k+1}\geq u_{k+1}] \leq \Pr\brk{\xi_k \leq M(1-\alpha^{k-1} + \alpha^k)} \leq \alpha^kM p_k \leq \alpha$.
\end{proof}

\begin{proof}[Proof of \Cref{lem:max_prob_bound}]
    We prove by induction that for every $k\in \{1,\dots,n\}$ and every $d_2,\dots,d_k$ with $n \leq d_i \leq u_i$, it holds
    \[
    \textstyle\Pr\brk{\xi_{\tau^{\star} \lor k} = \max_{i\in [n]} \xi_i \wedge W \midd \calE(d_2,\dots,d_k)} \leq \kappa(n) + 4(n-k)\alpha.
    \]
    First, note that for any $k$ we have
    $\{\xi_k = \max_i \xi_i\} \cap W \subseteq \{\xi_k = \xi_k^{\star}\}.$
    Now suppose $k=n$. Then
    \begin{align*}
    \textstyle\Pr\brk{\xi_{\tau^{\star} \lor n} = \max_{i\in [n]} \xi_i \wedge W \midd \calE(d_2,\dots,d_k)} &\leq \Pr\brk{\xi_{n} = \xi_n^{\star} \midd \calE(d_2,\dots,d_k)} \\ 
    &= \Pr[\xi_n = \xi_n^{\star}] = \kappa(n),
    \end{align*}
    where we used independence (\Cref{lem:secretary_independence}).
    Now let $k < n$ and suppose the statement holds true for $k+1,\dots,n$. Let $\nu =\Pr(\tau^{\star} = k\mid E(d_2,\dots,d_k))$. 
    We have
    \begin{align*}
       &\textstyle\Pr\brk{\xi_{\tau^{\star} \lor k} = \max_{i\in [n]} \xi_i \wedge W \midd \calE(d_2,\dots,d_k)} \\
       &= \textstyle\nu\Pr\brk{\xi_k = \max_{i\in [n]} \xi \wedge W \midd \calE(d_2,\dots,d_k)} + (1-\nu) \Pr\brk{\xi_{\tau^{\star} \lor k+1} = \max_{i\in [n]} \xi_i \wedge W \midd \calE(d_2,\dots,d_k)}.
    \end{align*}
    Now again, $\{\xi_k = \max_{i \in[n]} \xi_i\} \cap W \subseteq \{\xi_k = \xi^{\star}_k\}$, and by also using independence, we get
    \[
   \textstyle\Pr\brk{\xi_k = \max_{i\in [n]} \xi_i \wedge W \midd \calE(d_2,\dots,d_k)} \leq \Pr\brk{\xi_k = \xi_k^{\star} \midd \calE(d_2,\dots,d_k)} =  \kappa(n).
    \]
On the other hand,
\begin{align*}
&\textstyle\Pr\brk{\xi_{\tau^{\star} \lor (k+1)} =\max_{i\in [n]} \xi_i \wedge \midd D_i = d_i \text{ for all } i\in \{2,\dots,k\}} \\
    &= \sum_{n \leq d_{k+1} \leq u_i} \textstyle \Pr\brk{\xi_{\tau^{\star} \lor (k+1)} = \max_{i\in [n]} \xi_i \wedge W \midd \calE(d_2,\dots,d_{k+1})} \cdot \Pr\brk{D_{k+1}= d_{k+1} \midd \calE(d_2,\dots,d_k)} \\ 
    &\quad+ \textstyle \Pr\brk{\xi_{\tau^{\star} \lor (k+1)} = \max_{i\in [n]} \xi_i \wedge W \wedge D_{k+1} < 0 \midd \calE(d_2,\dots,d_k)} \\ 
    &\quad+ \textstyle \Pr\brk{\xi_{\tau^{\star} \lor (k+1)} = \max_{i\in [n]} \xi_i \wedge W \wedge 0 \leq D_{k+1} \leq n \midd \calE(d_2,\dots,d_k)} \\ 
    &\quad+ \textstyle \Pr\brk{\xi_{\tau^{\star} \lor (k+1)} = \max_{i\in [n]} \xi_i \wedge W \wedge u_{k+1} \leq D_{k+1} \midd \calE(d_2,\dots,d_k)}.
\end{align*}
By the induction hypothesis, the first term is at most $\kappa(n) + 4(n-k-1)\alpha$. Regarding the second term, note that $\xi_{\tau \lor (k+1)} =\max_{i\in [n]} \xi_i$ implies that $\xi_k < \xi_k^{\star}$. 
Hence, by independence and \Cref{lem:informative_events_unlikely}, it is bounded above by $\alpha$. 
By the same lemma, we can bound the third and fourth terms by $2\alpha$, and $\alpha$, respectively, so we get
\begin{align*}
&\textstyle \Pr\brk{\xi_{\tau^{\star} \lor (k+1)} = \max_{i\in [n]} \xi_i \wedge W \midd \calE(d_2,\dots,d_k)} \\
&\leq \kappa(n) + 4(n-k-1)\alpha + 4\alpha = \kappa(n) + 4(n-k)\alpha,
\end{align*}
which concludes the proof of the lemma.
\end{proof}

\section{Beyond Structured Correlations}\label{sec:arbitrary_correlations}

In this section, we turn our attention to settings that allow for arbitrary correlations between the random variables, but we parameterize the {\it degree} of dependency. 
We present tight asymptotic guarantees on the competitive ratios of the prophet and the secretary problems, depending on the chromatic number $\chi$ of the independence graph. By Brooks' Theorem \citep{brooks-theorem}, this further allows us to estimate the competitive ratio in terms of $\Delta$, the maximum degree of the independence graph. Specifically, we present an algorithmic scheme that uses an algorithm $\calB$ for pairwise-independent random variables as a subroutine. Using the $({1}/{3})$-competitive algorithm of~\citet[Theorem 2]{caragiannis-pairwise} for pairwise-independent random variables, we obtain an $O({{1}/{\chi}})$-competitive algorithm for arbitrary correlated instances. For the secretary objective, we present an algorithm for pairwise-independent random variables (\Cref{prp:secretary-pairwise-indep}) which selects the maximum realization with probability ${1}/{4}$. We then use this algorithm as a subroutine in our algorithmic scheme to obtain an algorithm for arbitrarily correlated instances that selects the maximum realization with probability $O({{1}/{\chi}})$.

Before we present the algorithm, we define the notion of independence graph that we use.
For an instance $I = \set{Y_1, \dots, Y_n}$ of the prophet inequality, the {\it independence graph} of $I$, denoted by $G_I = (V, E)$, is a graph where each vertex $v_i \in V$ corresponds to a random variable $Y_i$ and there exists an edge between two vertices $v_i, v_j$ if and only if $Y_i, Y_j$ are not pairwise-independent.

\begin{algorithm}[t]
\caption{Graph-based scheme for correlated instances}
\label{alg:graph-scheme}
Let $G_I=(V,E)$ be the independence graph of $I$, and
$\bm{S}$ a partition of $V$ into independent sets\;\\
Let $S^\star$ be an independent set drawn uniformly at random from $\bm{S}$\;\\
Instantiate subroutine $\calB$ for $|S^\star|$ elements\;

\For{$i\gets 1$ \KwTo $n$}{
  \If{$Y_i\in S^\star$}{
    Feed $Y_i$ to $\calB$\;\\
    \If{$\calB$ accepts $Y_i$}{
      Accept $Y_i$\;
    }
  }
}
\end{algorithm}

\begin{lemma}\label{lem:graph-cr-pbm}
Let $\calB_1$ be an $\alpha$-competitive algorithm in the prophet objective for pairwise-independent random variables and $\calB_2$ be a $\beta$-competitive algorithm in the secretary objective for pairwise-independent random variables. For every instance $I$ of the prophet inequality problem with correlated random variables, let $\chi(G_I)$ denote the chromatic number of $G_I$. Then, \Cref{alg:graph-scheme} with subroutine $\calB_1$ selects a value $Y_{\tau, 1}$ such that
${\E[{Y_{\tau, 1}}]}/{\E[{\max_{i \in [n]} Y_i}]} \geq {\alpha}/{\chi(G_I)},$
and with subroutine $\calB_2$ selects a value $Y_{\tau, 2}$ such that
$\Pr\brk{Y_{\tau, 2} = \max_{i \in [n]} Y_i} \geq {\beta}/{\chi(G_I)}.$
\end{lemma}
\begin{proof}
Since the chromatic number of $G_I$ is $\chi(G_I)$, we can partition $G_I$ into $\chi(G_I)$ independent sets. Notice that the random variables corresponding to vertices in the same independence set are pairwise independent since no two such vertices share an edge.

Let $S_{\textsc{max}}$ denote the (random) set from $\bm{S}$ that contains the maximum realization. Since $S^*$ is chosen uniformly at random, we have that
\begin{equation}\label{eq:pairwise-1}
\Pr[S^* = S_{\textsc{max}}] \geq \frac{1}{\chi(G_I)}.
\end{equation}
By assumption, our subroutine $\calB_1$ selects a value $Y_{\tau, 1}$ such that
\begin{equation}\label{eq:pairwise-2}
\E\brk{Y_{\tau, 1}} \geq \alpha\cdot \E\brk{\max_{i \in S^*} Y_i},
\end{equation}
and our subroutine $\calB_2$ returns a value $Y_{\tau, 2}$ such that
\begin{equation}\label{eq:pairwise-3}
\Pr\brk{Y_{\tau, 2} = \max_{i \in S^*} Y_i} \geq \beta.
\end{equation}
Combining \eqref{eq:pairwise-1} with \eqref{eq:pairwise-2} and \eqref{eq:pairwise-3}, we get
\[
\E\brk{Y_{\tau, 1}} \geq \alpha \cdot \E\brk{\max_{i \in S^*} Y_i} \geq \alpha \cdot \Pr[S^* = S_{\textsc{max}}] \cdot \E\brk{\max_{i \in S^*} Y_i \midd S^* = S_{\textsc{max}}} \geq \frac{\alpha}{\chi(G_I)}\cdot \E\brk{\max_{i \in [n]} Y_i},
\]
and
\begin{align*}
\Pr\brk{Y_{\tau, 2} = \max_{i\in [n]} Y_i} &\geq \Pr[S^* = S_{\textsc{max}}] \cdot \Pr\brk{Y_{\tau, 2} = \max_{i \in [n]} Y_i \midd S^* = S_{\textsc{max}}} \\
&\geq \Pr[S^* = S_{\textsc{max}}] \cdot \Pr\brk{Y_{\tau, 2} = \max_{i \in S^*} Y_i \midd S^* = S_{\textsc{max}}}\geq \frac{\beta}{\chi(G_I)}.\qedhere
\end{align*}
\end{proof}

For the prophet objective, we can use the $\f{1}{3}$-competitive algorithm by~\citet{caragiannis-pairwise} for pairwise-independent random variables. However, we must design an algorithm for pairwise-independent instances from scratch for the secretary setting.
\begin{proposition}\label{prp:secretary-pairwise-indep}
There exists a $0.25$-competitive algorithm for pairwise-independent random variables in the secretary objective.
\end{proposition}
\begin{proof}
For a fixed threshold $T$, let $p_i = \Pr[Y_i \geq T]$,  for every $i\in [n]$. 
Consider the algorithm $\calA$ that sets a threshold $T$ such that
$\sum_{i = 1}^n {p_i} = {1}/{2}.$
We prove that $\Pr[\calA \text{ selects } \max_{i\in [n]} Y_i] \geq \f{1}{4}.$
Notice that $\calA$ always selects the maximum realization if there is only one realization above $T$. Therefore
\begin{align*}
\Pr\brk{\calA \text{ selects } \max_{i\in [n]} Y_i} &\geq \sum_{i = 1}^n {\Pr[Y_i \geq T \wedge \text{for all } j \neq i, \: Y_j < T]} \\
&= \sum_{i = 1}^n {\Pr[Y_i \geq T] \cdot \Pr\brk{\text{for all } j \neq i, \: Y_j < T \midd Y_i \geq T}} \\
&\geq \sum_{i = 1}^n {p_i \cdot \left(1 - \sum_{j \neq i} \Pr\brk{Y_j \geq T \midd Y_i \geq T}\right)} \\
&= \sum_{i = 1}^n {p_i \cdot \left(1 - \sum_{j \neq i} p_j\right)}\geq \frac{1}{2} \sum_{i = 1}^n {p_i}= \frac{1}{4},
\end{align*}
where the second inequality follows from a simple union bound and the second equality from the fact that $Y_i, Y_j$ are pairwise independent.
\end{proof}
Since the chromatic number is at most $\Delta + 1$ for all graphs, where $\Delta$ is the maximum degree of a vertex in $G_I$, \Cref{lem:graph-cr-pbm} together with \Cref{prp:secretary-pairwise-indep} and the algorithm by~\citet{caragiannis-pairwise} directly imply the following result.
\begin{theorem}\label{thm:graph-degree-bound}
There exist algorithms $\calA_1$ for the prophet objective and $\calA_2$ for the secretary objective, such that for every instance $I$ with correlated random variables where $\max_{v \in V(G_I)} {\deg(v)} \leq \Delta$, they return values $Y_{\tau, 1}$ and $Y_{\tau, 2}$, respectively, such that
\[
\frac{\E\brk{Y_{\tau, 1}}}{\E\brk{\max_{i \in [n]} Y_i}} \geq \frac{1}{3(\Delta+1)},\text{ and }\Pr\brk{Y_{\tau, 2} = \max_{i\in [n]} Y_i} \geq \frac{1}{4 (\Delta+1)}.
\]
\end{theorem}

\section{Discussion}\label{sec:beyond-linear-discussion}

We take a novel approach to the classic $1/2$-competitive ratio guarantee for independent random variables, via a differential equation, that unifies many previously known proofs. We believe this new differential equation approach could be useful in several prophet inequality applications. To illustrate this, we apply this new understanding to investigate the properties of prophet inequality instances with correlations; specifically, instances where the observed random variables are a function of some independent random variables $X_1, \dots, X_m$. On one hand, if $Y_i = X_i + Z$ for independent $X_i$'s and $Z$, we make progress in understanding the exact constant factor achievable. On the other hand, we provide a strong negative result showing that, if the function of the $Y_i$'s is allowed to contain products of two random variables, no non-trivial guarantees are possible.

Our paper also paves the way for further study of correlated instances under assumptions about the degree of correlation, and we believe that a better understanding of how this degree affects the guarantees is an interesting direction for future work. Since our results preclude the $Y_i$'s from involving products of random variables, it is an interesting question to better understand linear correlations and the conditions under which the $1/2$ factor continues to hold.

\paragraph{Acknowledgments.} J. Correa was partially funded by the Center for Mathematical Modeling (ANID grant FB210005). R. Jlibene and R. Laraki were supported by the IRP Fair Games of the CNRS. V. Livanos was partially funded by the Air Force Office of Scientific Research under award number FA9550-24-1-0261. V. Verdugo was partially funded by ANID FONDECYT (grant 1241846).

\bibliographystyle{abbrvnat}
\bibliography{references}

\end{document}